\documentclass[12pt]{iopart}

\usepackage{times}

\usepackage[utf8x]{inputenc}
\usepackage[english]{babel}
\usepackage[T1]{fontenc}
\usepackage{lmodern}
\usepackage{amssymb}
\usepackage{bbold}
\usepackage{amsthm}
\usepackage{amsfonts}
\usepackage[pdftex]{graphicx}
\usepackage{wrapfig}
\usepackage{epstopdf}

\usepackage{tabularx}

\usepackage{dcolumn}
\usepackage{bm}
\usepackage{color}
\usepackage{etoolbox}
\usepackage{subcaption}

\usepackage{pgfplots}

\usepackage{lipsum}
\usepackage{ulem}
\usepackage{bbold}
\usepackage{bm}
\usepackage{braket}
\DeclareUnicodeCharacter{2010}{-}
\usepackage{cancel}
\usepackage{color}
\usepackage{stackrel}

\newcommand{\Trace}{\mathrm{tr}}
\newcommand{\Kol}{\mathrm{kol}}

\renewcommand{\Tr}[1]{\mathrm{tr}\left[#1\right]}          

\newcommand{\Exp}[1]{\exp \left[ #1 \right]}  

\newcommand{\Span}[1]{\mathrm{span} \left\lbrace #1\right\rbrace}
\newcommand{\dis}{\mathrm{d}}  

\newcommand{\1}{\mathbb{1}}      

\newcommand{\HS}{\mathrm{HS}}    
\newcommand{\SIC}{\mathrm{SIC}}  

\newcommand{\with}{\mathrm{with}}
\newcommand{\und}{\mathrm{and}}

\newcommand{\Hil}{\mathcal{H}}  
\newcommand{\Evl}{\mathcal{V}}  
\newcommand{\E}{\mathcal{E}}    
\newcommand{\B}{\mathcal{B}}    
\newcommand{\F}{\mathcal{F}}    
\newcommand{\A}{\bm{a}}   
\newcommand{\q}{\bm{q}}   
\newcommand{\p}{\bm{p}}   
\newcommand{\x}{\vec{x}}  
\newcommand{\X}{\vec{\bm{X}}}    
\newcommand{\reals}{\mathbb{R}}  
\newcommand{\var}[1]{\mathrm{var}\left(#1\right)}  

\newcommand{\Rho}{\bm{\rho}}  

\newcommand{\sigmabm}{\bm{\sigma}}  
\newcommand{\Deltabm}{\bm{\Delta}}  
\newcommand{\Omegabm}{\bm{\Omega}}  

\renewcommand{\e}{\mathbb{e}}  

\newcommand{\ketbra}[2]{|#1\rangle\langle#2|} 

\newcommand{\ubs}[2]{\stackrel[#1]{}{\underbrace{#2}}}  
\newcommand{\underset}[2]{\stackrel[#1]{}{#2}}          

\newtheorem{thm}{Theorem}

\newtheorem{definition}{Definition}

\begin{document}

\title{Witnessing non-Markovianity by Quantum Quasi-Probability Distributions}

\author{Moritz F. Richter¹, Raphael Wiedenmann¹ and Heinz-Peter Breuer¹²}

\address{¹ Physikalisches Institut, Universit\"at Freiburg, 
Hermann-Herder-Stra{\ss}e 3, D-79104 Freiburg, Germany}
\address{² EUCOR Centre for Quantum Science and Quantum Computing,
University of Freiburg, Hermann-Herder-Stra{\ss}e 3, D-79104 Freiburg, Germany}
\vspace{10pt}
\begin{indented}
\item[]October 2022
\end{indented}

\begin{abstract}
We employ frames consisting of rank-one projectors (i.e.~pure quantum states) and their induced informationally complete quantum measurements (IC-POVMs) to represent generally mixed quantum states by quasi-probability distributions. In the case of discrete frames on finite dimensional systems this results in a vector like representation by quasi-probability vectors, while for the continuous frame of coherent states in continuous variable (CV) systems the approach directly leads to the celebrated representation by Glauber-Sudarshan P-functions and Husimi Q-functions. We explain that the Kolmogorov distances between these quasi-probability distributions lead to upper and lower bounds of the trace distance which measures the distinguishability of quantum states. We apply these results to the dynamics of open quantum systems and construct a non-Markovianity witness based on the Kolmogorov distance of the P- and Q-functions. By means of several examples we discuss the performance of this witness and demonstrate that it is useful in the regime of high entropy states for which a direct evaluation of the trace distance is typically very demanding. For Gaussian dynamics in CV systems we even find a suitable non-Markovianity measure based on the Kolmogorov distance between the P-functions which can alternatively be used as a witness for non-Gaussianity.
\end{abstract}

\noindent{\it Keywords\/}: Open quantum systems, non-Markovian dynamics, quantum measurement, IC-POVM, frames, quasi-probability distributions, Kolmogorov distance

\maketitle

\section{Introduction}
\label{Introduction}

A topical issue in the quantum theory of open systems \cite{Breuer2007} is the definition, detection and quantification of non-Markovian dynamics, i.e. dynamical processes determined by memory effects \cite{Rivas2014a,Breuer2016,deVega2017}. One approach to the definition of non-Markovianity in the quantum regime is based on the concept of an information flow between the open system and its environment \cite{Breuer2009b}. According to this concept Markovian dynamics is characterized by a continuous flow of information from the open system to its environment, while non-Markovian behavior features a flow of information from the environment back to the open system. To formulate these ideas in mathematical terms one has to introduce a suitable measure for the information inside the open system. A natural such measure is given by the trace distance \cite{Nielsen2000} between pairs of open system states, since this distance has a direct interpretation in terms of the distinguishability  of the quantum states \cite{Hayashi2006}. Thus, a quantum process is said to be Markovian if the trace distance and, hence, the distinguishablility decreases monotonically in time and, vice versa, it is non-Markovian if the trace distance shows a non-monotonic behavior, implying a temporary increase of the distinguishability which characterizes the information backflow. We note that there are also entropic measures for this information with similar properties \cite{Megier2021,Settimo2022}. As a further concept one introduces the dynamical maps $\Lambda_t$, $t \geq 0$, which describe the evolution in time of the density matrix $\Rho$ of the open system as $\Rho(t)=\Lambda_t[\Rho(0)]$. This map is known to be completely positive and trace preserving (CPTP) and has the important property to be a contraction for the trace distance \cite{Ruskai1994a,Nielsen2000}. This leads to the following measure for the degree of memory effects of a certain pair of initial states $\Rho_1$ and $\Rho_2$ of the open system \cite{Breuer2009b,Laine2010a},
\begin{eqnarray} \label{NM-measure}
\mathcal{N}(\Rho_1, \Rho_2, \{\Lambda_t\}) := \underset{\sigma \geq 0}{\int} dt \ \sigma(\Rho_1, \Rho_2, t) \\
\with \quad \sigma(\Rho_1, \Rho_2, t) := \frac{d}{dt} \dis_\Trace\big(\Lambda_t[\Rho_1], \Lambda_t[\Rho_2] \big),
\end{eqnarray}
quantifying the total backflow of information form the environment into the system. Obviously, this quantity depends on the chosen pair of initial states. To make this quantity independent of the initial state pair and to obtain a function which only depends on the family of dynamical maps, one can take the maximum over all initial state pair leading to the non-Markovianity measure
\begin{eqnarray} \label{NM-measure-max}
\mathcal{N}(\{ \Lambda_t \}) \ := \ \underset{\Rho_1, \Rho_2}{\mathrm{max}} \mathcal{N}(\Rho_1, \Rho_2, \{ \Lambda_t \}).
\end{eqnarray}

This information flow approach to quantum non-Markovianity has been applied to many theoretical models and experimental systems (see, e.g., the reviews \cite{Rivas2014a,Breuer2016} and references therein). However, most applications deal with relatively simple open systems, describing for example a qubit or a small number of interacting qubits (e.g. a spin-chain) linearly coupled to bosonic or fermionic reservoirs with a structured spectral density. For more complicated systems not only the determination of the dynamical evolution is much more involved, even if it can be represented in terms of an exact or approximate master equation, but also the calculation of the trace distance between the quantum states can become highly demanding, especially for infinite dimensional \textit{continuous variable} (CV) systems. In this paper we want to tackle these difficulties by investigating non-Markovianity measures or at least non-Markovianity witnesses which are more suitable for CV systems, employing representations of quantum states by means of quasi-probability distributions.

To this end, we start in section \ref{Quantum Frames and IC-POVM} by introducing a vector-like representation of quantum states based on generalized bases on bounded operators consisting of a set of fixed but not mutually orthogonal pure states which we will address as \textit{quantum frames}. We show that for any quantum state its decomposition into such a quantum frame can be understood as a quasi-probabilistic mixture, and we will recapitulate how one can connect these quantum frames to so called \textit{informationally complete positive operator valued measure} or IC-POVM \cite{Renes2004-PhD, Renes2004}, a class of quantum measurements especially useful in quantum tomography and even experientially realizable \cite{Bian2015, Bent2015}. The frame based and the IC-POVM based representation of a quantum state can then be used in section \ref{Distance Measures on Quantum States} to define distance measures for quantum states by applying the Kolmogorov distance for probability distribution. We then will show how to use these measures in order to approximate the trace distance.

While section \ref{Quantum Frames and IC-POVM} considers finite dimensional quantum systems only, we will reformulate in section \ref{Coherent States as Continuous Quantum Frame} the idea for the continuous quantum frame of coherent states in CV systems and its connection to the Glauber-Sudarshan P-function and the Husimi Q-function \cite{Serafini2017}. Again we apply the Kolmogorov distance to these quasi-probability distributions in order to approximate the trace distance between given quantum states. We also discuss the performance of this approximation with the help of a class of randomly generated Gaussian states \cite{Adesso2014}. On the basis of this approximation we will construct in section \ref{Non-Markovianity and CV-Systems} a suitable witness for the non-Markovianity in CV systems using P- and Q-functions (subsection \ref{Witnessing non-Markovianity with Q and P-functions}) and illustrate this by means of the example of a non-Markovian quantum oscillator (subsection \ref{Example: Non-Markovian damped oscillator}). In subsection \ref{Gaussian Evolution of P-Function Kolmogorov Distance} we focus on a special class of dynamical maps which preserve the Gaussianity of an initial quantum state and discuss the time evolution of the Kolmogorov distance between P-functions under those maps. This will lead us to a novel measure for non-Markovianity for the case of Gaussian dynamics, since the Kolmogorov distance turns out to be contracting under Gaussian CPTP maps. Finally we draw our conclusions in section \ref{Conclusion} and give an outlook on possible future studies, where one could apply the witnesses and measures of non-Markovianity developed here.

\section{Quantum Frames and IC-POVM}
\label{Quantum Frames and IC-POVM}
We begin with a brief introduction to frames in general Hilbert spaces in order to clarify the terminology and to sketch basic concepts. For more details, the reader is referred to Ref.~\cite{Kovacevic2008}. A \textit{frame} $\F = \{\ket{f_i}\} \subset \Hil$ in a Hilbert space $\Hil$ is a set of vectors satisfying $\Span{\F} = \Hil$. The immediate question now is: What is the difference to a basis in $\Hil$? The answer is the following:
\begin{itemize}
\item Firstly, a basis has to be a minimal set spanning $\Hil$ while a frame might be overcomplete. In this sense a frame is just a generalized notion of a basis.
\item Secondly, when speaking of a \textit{basis} we usually mean an \textit{orthonormal} basis (ONB) for which $\braket{f_i|f_j} = \delta_{ij}$. Although mathematically a basis is neither required to be orthogonal nor normalized, we will use the term \textit{frame} to stress that orthonormality of its elements is not assumed.
\end{itemize}
The possible non-orthogonality causes the major technical differences between ONBs and frames. An important object in frame theory is the frame operator $S: \Hil \to \Hil$ defined by
\begin{equation}
 S \ket{v} = \sum_i \braket{f_i|v}\ket{f_i}.
\end{equation}
In case of an ONB the frame operator is simply the identity operator while for general frames it is not. Later on, this frame operator will provide the link between a frame induced decomposition and a frame induced measurement of a state. Furthermore, we can define a \textit{canonical dual frame} via the frame operator as
\begin{eqnarray}
\tilde{\F} := \{\ket{\tilde{f}_i} = S^{-1} \ket{f_i} \} \qquad \with \quad \ket{v} = \sum_i \braket{\tilde{f}_i| v} \ket{f_i}.
\end{eqnarray}
Again, in the case of an ONB we find $\tilde{\F} = \F$ and, hence, an ONB can be addressed as a \textit{minimal and self-dual frame}. However, although orthogonality simplifies decomposition and spanning significantly, it might be of advantage to use non-orthogonal bases with certain other useful properties instead  as we will do in this paper (see below). Even further, it might be desirable to have a spanning set, i.e. a frame, which covers some part of the vector space at hand more densely such that vectors one wants to decompose are already closer (in a certain sense) to some frame elements and have decompositions more pronounced on a single such element.

Let us now consider a quantum system with finite-dimensional Hilbert space $\Hil$. We denote the space of bounded operators on $\Hil$ by $\B(\Hil)$, which becomes a Hilbert spaces on its own via the Hilbert-Schmidt scalar product
\begin{eqnarray}
\braket{A,B}_\HS := \Tr{A^\dagger B}.
\end{eqnarray}
The physical states of the quantum system are represented by density matrices $\Rho$ which are Hermitian and positive operators with trace one \cite{Nielsen2000}.  Pure states  of the system are represented by rank-one projections $\Rho=\ketbra{\varphi}{\varphi}$, where $\ket{\varphi} \in \Hil$ is a normalized state vector. Any density operator has a spectral decomposition of the form $\Rho = \sum_i p_i \ketbra{i}{i}$ with $\vec{p}=(..., p_i, ...)^T$ a probability distribution and $\{\ket{i}\}$ an ONB in $\Hil$. Note that two different mixed quantum states not only differ in their probability distribution $\vec{p}$ but typically have different ONBs in their spectral decomposition, too. By contrast, the concept of frames allows the following definition which will give a decomposition into a fixed set of pure states.
\begin{definition}
A frame $\F := \{\ketbra{\psi_i}{\psi_i} \} \subset \B(\Hil)$ consisting of pure quantum states is called a \textbf{quantum frame}.
\end{definition}
In the following we will assume $\F$ to be minimal (i.e.~the $\ketbra{\psi_i}{\psi_i}$ form a basis in $\B(\Hil)$ which, however, cannot be orthogonal). Accordingly, any density matrix can be decomposed as
\begin{equation}
 \Rho = \sum_i f_i \ketbra{\psi_i}{\psi_i},
\end{equation}
where we have
\begin{eqnarray}
\Rho^\dagger = \Rho  &\Leftrightarrow&  \ f_i \in \mathbb{R} \ \forall i \\
\Tr{\Rho} = 1 \ &\Leftrightarrow& \ \sum_i f_i = 1 \\
\braket{\varphi|\Rho|\varphi} \geq 0 \ &\Leftrightarrow& \ \sum_i f_i |\braket{\psi_i|\varphi}|^2 \geq 0.
\end{eqnarray}
Note that the last line has to be fulfilled for all $\ket{\varphi}\in\Hil$ in order to ensure the positivity of $\Rho$. The coefficients $f_i$ are referred to as \textit{frame decomposition coefficients} (FDCs). We see that using quantum frames we can decompose any mixed quantum state quasi-probabilistically (i.e.~all $f_i$ are real and sum up to one but might be negative) into a \textit{fixed} set of pure states. Furthermore, $\vec{f} = (\dots, f_i, \dots )^T$ is a vector like representation of the mixed quantum state at hand which we will call the \textit{frame vector} and indicate by means of
\begin{equation}
\vec{f} \simeq \Rho
\end{equation}
the one-to-one correspondence between the frame vector and the density matrix. The frame operator is represented by
\begin{equation}
S_{ij} = \Tr{\ketbra{\psi_i}{\psi_i} \ketbra{\psi_j}{\psi_j} } = |\braket{\psi_i|\psi_j}|^2,
\end{equation}
where $\bm{S} \vec{f}$ is the frame vector of $S(\Rho)$.

\begin{figure}
\center
\includegraphics[width = 5.5cm]{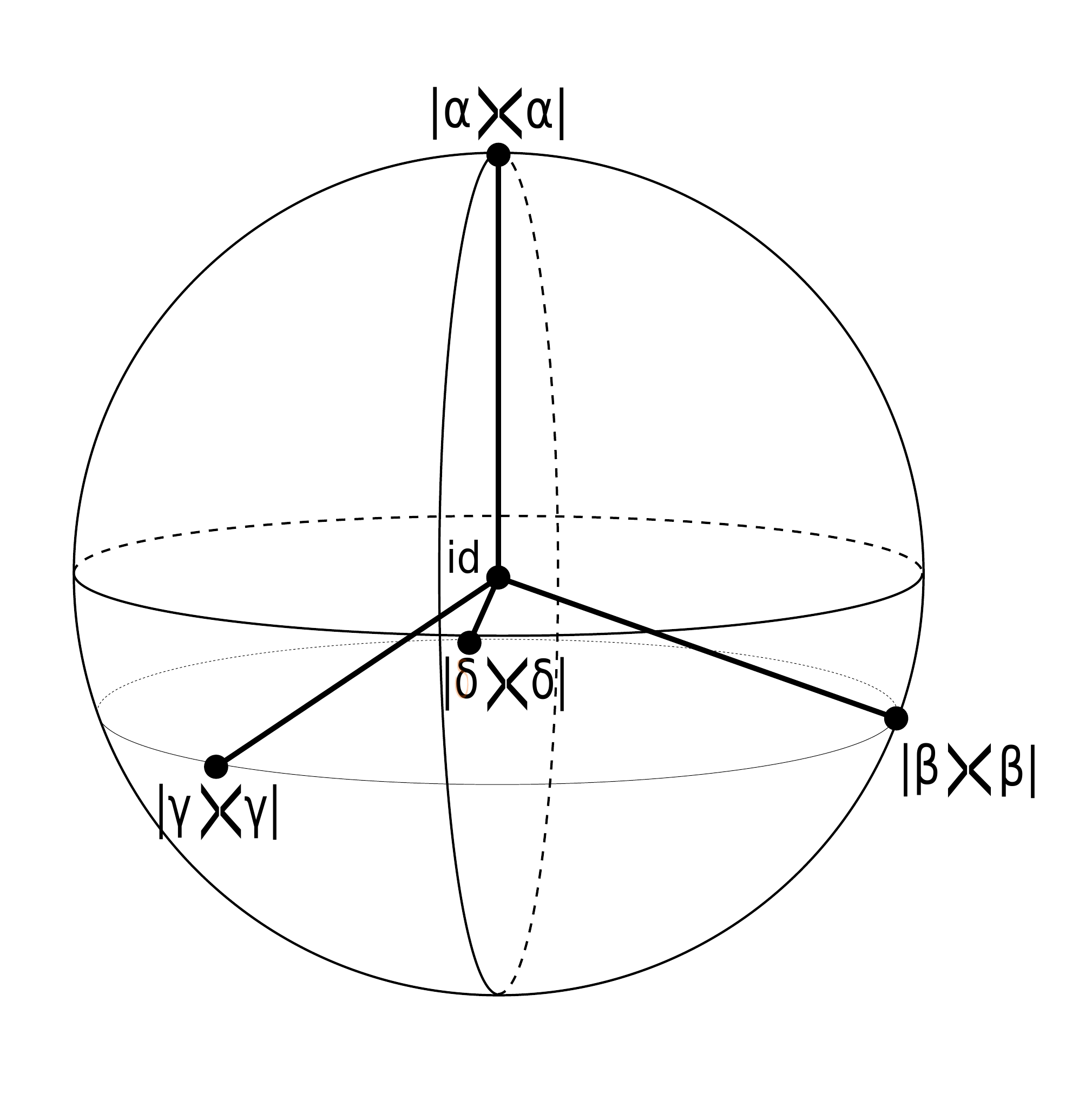}
\caption{\begin{footnotesize} Example of the quantum frame $\F_\mathrm{SIC} = \{ \ketbra{\alpha}{\alpha}, \ketbra{\beta}{\beta}, \ketbra{\gamma}{\gamma}, \ketbra{\delta}{\delta} \}$ inducing a \textit{symmetric-IC} or \textit{SIC-POVM} \cite{Renes2004} on a qubit ($\dim \Hil = 2$) represented by the Bloch sphere. The frame is minimal (i.e. can be regarded as a non-orthonormal basis of $\B(\Hil)$) and due to its symmetry takes the shape of a regular tetrahedron in the Bloch sphere.\end{footnotesize}}
\label{fig:SIC-2D}
\end{figure}

A given quantum frame $\F = \{\ketbra{\psi_i}{\psi_i} \}$ can also be associated to a rank-one \textit{informationally complete positive operator valued measure} or \textit{IC-POVM} given by $\{\E_i := \frac{1}{c_i} \ketbra{\psi_i}{\psi_i} \}$ such that $\sum_i \E_i = \1$  \cite{Renes2004, Renes2004-PhD}. Such an IC-POVM is a mathematical expression of a tomographic quantum measurement mapping a quantum state $\Rho$ to a probability distribution $\vec{p}$ with \cite{Nielsen2000}
\begin{equation}
p_i = \Tr{\E_i \Rho} = \sum_j f_j \frac{1}{c_i} \Tr{\ketbra{\psi_i}{\psi_i} \ketbra{\psi_j}{\psi_j}} = \sum_j \frac{1}{c_i} S_{ij} f_j =: \sum_j M_{ij} f_j,
\label{eq:Mf=p}
\end{equation}
where we have introduced the matrix $\bm{M}$ with elements $M_{ij}=\frac{1}{c_i} S_{ij}$. This means that $\E_i$ is the \textit{effect operator} of outcome $i$ while $p_i$ is the probability to measure this outcome if the system at hand is prepared in state $\Rho$. The term \textit{informationally complete} refers to the fact that due to the underlying frame structure different quantum states always have different IC-POVM probability distributions as well, i.e. the vector $\vec{p}$ encodes complete information about its quantum state $\Rho$ and hence is a vector like representation of it, too,
\begin{equation}
\vec{p} \simeq \Rho,
\end{equation}
as used e.g. in \cite{Kiktenko2020, Yashin2020} (indeed, it is the decomposition into the canonical dual frame of $\F$ which does not consist of pure quantum states anymore). Assuming a minimal quantum frame the matrix $\bm{M}$ is bijective and we find
\begin{equation}
\bm{M}\vec{f} = \vec{p} \qquad \und \qquad \vec{f} = \bm{M}^{-1} \vec{p}.
\end{equation}
One can understand the frame vector $\vec{f}$ as a decomposition based representation of $\Rho$ since it tells us how to construct the quantum state at hand using the given frame. Yet the IC-POVM probability vector $\vec{p}$ is a measurement based representation and encodes the results of measuring $\Rho$ using the frame-corresponding IC-POVM and the link between both representations is the frame operator $S$ reshaped as the matrix $\bm{M}$, as can be seen from Eq.~(\ref{eq:Mf=p}). Again, the difference between both representations is caused by the non-orthogonality of quantum frames.

A prominent and even experimentally realized example (see, e.g., \cite{Bian2015,Bent2015}) for quantum frame based IC-POVMs are \textit{symmetric}-IC-POVMs (or \textit{SIC-POVMs}) which are associated to a minimal frame with $|\braket{\psi_i|\psi_j}|^2 = \frac{1 + n \cdot \delta_{ij}}{1+n}$ and $c_i = n \ \forall i$, where $\dim \Hil = n$ \cite{Renes2004}, see Fig.~\ref{fig:SIC-2D}.

\section{Distance Measures for Quantum States}
\label{Distance Measures on Quantum States}
In section \ref{Quantum Frames and IC-POVM} we have introduced a frame based quasi-probabilistic decomposition of density matrices into a fixed set of pure states forming a so-called quantum frame. In this section we use these representations to introduce distance measures for quantum states and compare them with the trace distance which is known as a measure for the distinguishability of quantum states.

The trace distance between two quantum states $\Rho$ and $\sigmabm$ defined by \cite{Nielsen2000,Hayashi2006}
\begin{equation} \label{eq-trace-distance-1}
\dis_\Trace (\Rho, \sigmabm) :=  \frac{1}{2} ||\Rho - \sigmabm || \equiv \frac{1}{2}\Trace |\Rho - \sigmabm|,
\end{equation}
where $||A||$ denotes the trace norm of an operator $A$, while the modulus is given by $|A| = \sqrt{A^\dagger A}$ \cite{Reed1972}. Defining the selfadjoint operator $\Deltabm = \Rho - \sigmabm$ and introducing its spectral decomposition $\Deltabm = \sum_i \lambda_i \ketbra{i}{i}$, we can thus write the trace distance as a sum over the moduli of the eigenvalues $\lambda_i$ of $\Deltabm$,
\begin{equation} \label{eq-trace-distance-2}
\dis_\Trace (\Rho, \sigmabm) =  \frac{1}{2} \sum_i |\lambda_i|.
\end{equation}
The trace distance between two quantum states $\Rho$ and $\sigmabm$ can be shown to be a measure for the distinguishability of these states \cite{Nielsen2000,Hayashi2006} and has the following useful properties:
\begin{itemize}
\item $0 \leq \dis_\mathrm{tr}(\Rho, \sigmabm) \leq 1$ and $\dis_\mathrm{tr}(\Rho, \sigmabm) = 1 \ \Leftrightarrow \  \Rho \perp \sigmabm$
\item $\dis_\Trace(\Evl[\Rho], \Evl[\sigmabm]) = \dis_\Trace(\Rho, \sigmabm)$ for any unitary transformation $\Evl$
\item $\dis_\Trace(\Lambda[\Rho], \Lambda[\sigmabm]) \leq \dis_\Trace(\Rho, \sigmabm)$ for any \textit{positive trace preserving} (PTP) map $\Lambda$
\end{itemize}
We will use the trace distance later in section \ref{Non-Markovianity and CV-Systems} to quantify the information flow between an open system system and its environment and, hence, to quantify non-Markovianity.

According to Eq.~(\ref{eq-trace-distance-2}) we have to determine the spectrum of $\Deltabm = \Rho - \sigmabm$ in order to compute the trace distance between two quantum states $\Rho$ and $\sigmabm$. However, in many practical applications this can be a demanding problem, in particular in the case of high-dimensional or even infinite-dimensional Hilbert spaces. If, for example, the quantum states at hand represent mixtures of many basis states of a continuous variable system (see below), it can be a hard task to compute a full matrix representation of $\Deltabm$ and evaluate its eigenvalues.

One possible strategy to avoid these difficulties - which we will investigate in this paper - is to use the vector like frame and IC-POVM based representations of the quantum states under consideration. Since both are quasi-probabilistic one can apply distance measures for classical probability distributions to them, e.g.~the Kolmogorov distance \cite{Nielsen2000}
\begin{eqnarray}
\dis_\Kol(\vec{p}, \vec{q}) := \frac{1}{2}\sum_i |p_i - q_i|,
\end{eqnarray}
where $\vec{p}$ and $\vec{q}$ are two probability distributions. With these definitions one finds the following chain of inequalities,
\begin{eqnarray}
\dis_\Kol(\vec{p}_1, \vec{p}_2) \leq \dis_\Trace(\Rho_1, \Rho_2) \leq \dis_\Kol(\vec{f}_1, \vec{f}_2),
\label{eq:p-dis<tr-dis<f-dis}
\end{eqnarray}
which states that the trace distance between any pair of quantum states is bounded from below by the Kolmogorov distance of the corresponding quasi-probability vectors, and from above by the Kolmogorov distance between the corresponding frame vectors. The proof is rather straight forward and given in appendix \ref{Proof of Eq. (eq:p-dis<tr-dis<f-dis)}. Furthermore, please note that the left inequality can also be deduced from the more general statement \cite{Nielsen2000}
\begin{eqnarray}
\dis_\Trace(\Rho_1, \Rho_2) \ = \ \underset{\{\E_i\} \in \mathrm{POVMs}}{\mathrm{max}} \ \dis_\Kol(\vec{p}_{\{\E_i\}}(\Rho_1), \vec{p}_{\{\E_i\}}(\Rho_2))
\end{eqnarray}
where the maximum is taken over all possible POVMs (not just IC-POVMs) and $\vec{p}_{\{\E_i\}}(\Rho)$ means the probability distribution of a POVM $\{\E_i\}$ applied to a quantum state $\Rho$.

\begin{figure}
\includegraphics[width=0.95\columnwidth]{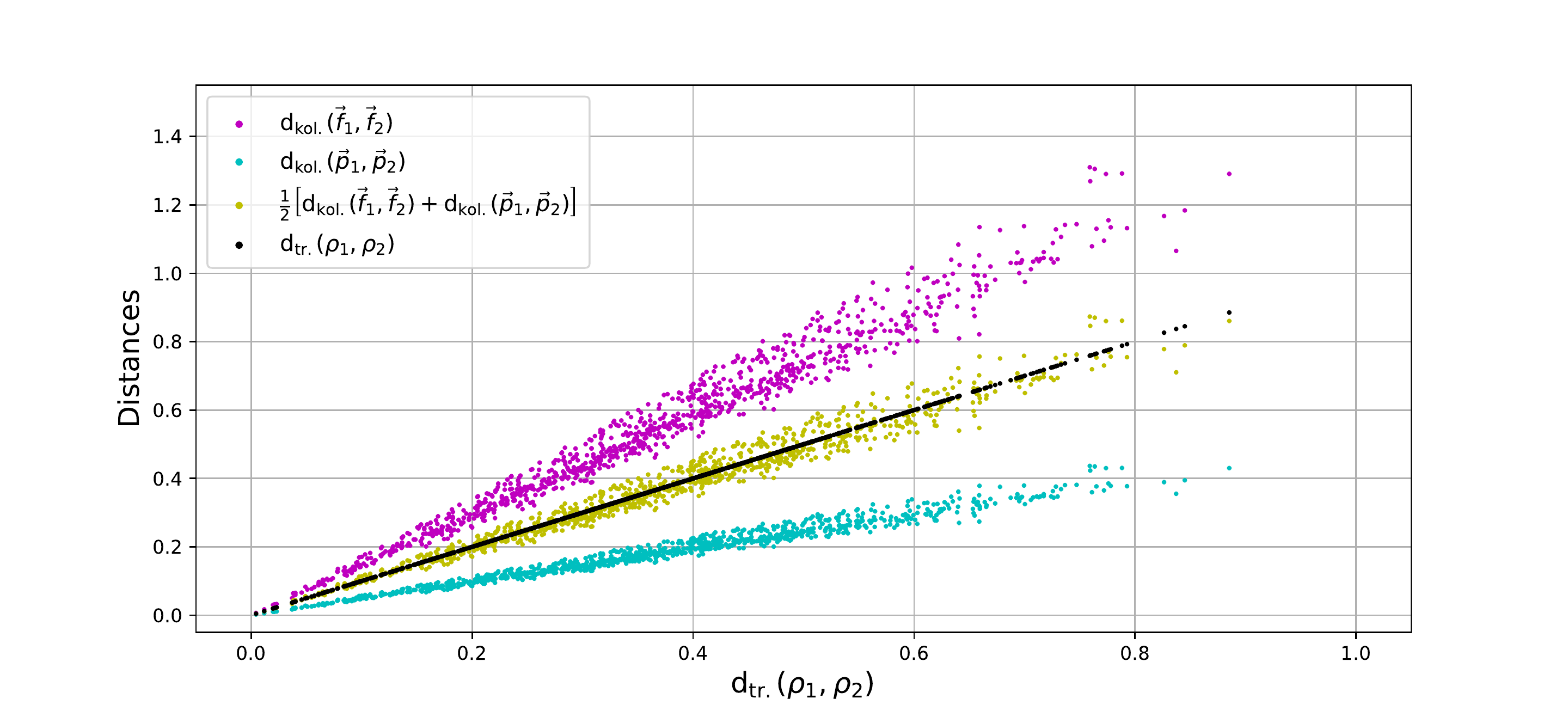}
\caption{\begin{footnotesize}A sample plot of 1000 randomly generated qubit state pairs $\Rho_1$ and $\Rho_2$. Plotted are the Kolmogorov distance between their frame vectors $\vec{f}_i$, their probability vectors $\vec{p}_i$ and the average $\frac{1}{2}(\dis_\Kol(\vec{f}_1, \vec{f}_2) + \dis_\Kol\big(\vec{p}_1, \vec{p}_2)\big)$ of both using the a symmetric minimal frame and its corresponding SIC-POVM over the trace distance $\dis_\Trace(\Rho_1,\Rho_2)$. The plot of trace distance over trace distance represented by the diagonal is shown for orientation as well. \end{footnotesize}}
\label{fig:Sample-Plot_Qubit-SIC.pdf}
\end{figure}

Figure \ref{fig:Sample-Plot_Qubit-SIC.pdf} illustrates the inequalities (\ref{eq:p-dis<tr-dis<f-dis}) for a SIC-POVM and its minimal and symmetric frame $\F_\SIC$ for a qubit. For randomly generated pairs of quantum states the different distance measures are plotted over the trace distance. First we see that in accordance with (\ref{eq:p-dis<tr-dis<f-dis}) the Kolmogorov distance between the frame vectors and between the SIC-POVM probability distributions - encoded in their probability vectors $\vec{p}$ - is allways above and below the trace distance, respectively. Furthermore, we find that with increasing trace distance the spreading of the Kolmogorov distances between frame vectors and between probability vectors increases as well. This is due to the fact that almost orthogonal state pairs allow for a larger variety of positioning relatively to the convex hull of the quantum frame $\F$ at hand. Either one of the two states is within the convex hull and its orthogonal complement is not, or both states are outside the convex hull. By definition states within the convex hull have frame vectors consisting of entries $f_i \in [0,1]$ only, while states outside the convex hull must have at least one negative entry and the sum of the other entries must be greater than one to compensate this negativity. Accordingly, the absolute values $|f_i|$ of the FDCs of a state inside the convex hull are on average smaller than in case of states outside the convex hull. Thus, if at least one state lies within the convex hull the Kolmogorov distance between the frame vectors will be smaller than in the case when both states are outside the convex hull.Further, this also explains that their is no general functional relation between trace distance and Kolmogorov distance between frame vectors since two pairs of quantum states might have the same trace distance but the Kolmogorov distance between their frame vectors might be different anyway. The same holds for probability vectors, though less pronounced since their Kolmogorov distances have to be smaller in general compared the distances between their corresponding frame vectors and hence, the probability vectors are closer to each other by construction. Interestingly, we find that the average of both Kolmogorov distances, i.e. $\frac{1}{2}\big(\dis_\Kol(\vec{f}_1, \vec{f}_2) + \dis_\Kol(\vec{p}_1, \vec{p}_2)\big)$, spreads as well with increasing trace distance but almost symmetrically around the trace distance.

\section{Coherent States as Continuous Quantum Frame}
\label{Coherent States as Continuous Quantum Frame}
We extend the concepts and investigations of section \ref{Distance Measures on Quantum States} to continuous variable systems, employing coherent states often used, e.g., in quantum optical systems, which actually form a continuous quantum frame \cite{Antoine2018} as will be explained below. We start with a brief summary of the features which are relevant for our purposes, further details may be found in Refs.~\cite{Klauder1985, A.Ferraro2005, Adesso2014, Serafini2017}.

Let us consider a \textit{continuous variable} (\textit{CV}) system, describing a set of continuous degrees of freedom of a quantum mechanical system, for example a harmonic oscillator or a bosonic  field mode. For simplicity we assume only a single mode but the concepts can easily be extended to multi-mode systems. As usual we denote the vacuum or ground state by $\ket{0}$ and the annihilation and creation operators by $\A$ and $\A^\dagger$, respectively. The Weyl displacement operator is given by $\bm{D}(\alpha) = \exp[\alpha \A^\dagger - \alpha^* \A]$ with $\alpha \in \mathbb{C}$ and $\alpha^*$ its complex conjugate. A coherent state is defined as \cite{Serafini2017}
\begin{eqnarray}
\ket{\alpha} := \bm{D}(\alpha)\ket{0} = \e^{-|\alpha|^2/2} \sum_{n=0}^\infty \frac{\alpha^n}{\sqrt{n!}} \ket{n}
\end{eqnarray}
with $\ket{n}$ the number states (Fock basis). The coherent state are eigenstates of the annihilation operator, $\A \ket{\alpha} = \alpha \ket{\alpha}$. Furthermore, the real and imaginary part of the complex number $\alpha$ describes the expectation values of position $q$ and of momentum $p$ of the coherent state. In fact, choosing appropriate units for position and momentum we have
\begin{eqnarray}
q = \sqrt{2} \ \Re(\alpha) \quad \und \quad p = \sqrt{2} \ \Im(\alpha).
\label{eq:sqrt(2)alpha=x}
\end{eqnarray}
Thus, we can regard the complex plane as a representation of the phase space.

Using (\ref{eq:sqrt(2)alpha=x}) we can express the Weyl displacement operator and coherent states in proper phase space coordinates $\x := (q,p)^T$. Let $\q$ and $\p$ be position and momentum operator, respectively, which are connected to creation and annihilation operators via the relations
\begin{eqnarray}
\q = \sqrt{2}(\A^\dagger + \A) \qquad \und \qquad \p = \sqrt{2}(\A^\dagger - \A).
\end{eqnarray}
With the help of the symplectic form of the phase space 
\begin{equation}
\Omegabm = \left(\begin{array}{cc} 0 & 1 \\ -1 & 0 \end{array}\right)
\end{equation} 
and $\X := (\q, \p)^T$ we can finally write \cite{Serafini2017}
\begin{eqnarray}
\bm{D}(-\x) := \bm{D}\big(\alpha(\x)\big) = \exp[-i \x^T \ \Omegabm \ \X],
\end{eqnarray}
where $\alpha(\x)$ simply means the complex number representing the point $\x$ in phase space as described by (\ref{eq:sqrt(2)alpha=x}). Accordingly, one can equivalently define coherent states as
\begin{eqnarray}
\ket{\x} := \ket{\alpha(\x)} = \e^{-|\x|^2/4} \sum_{n=0}^\infty \frac{(q + ip)^n}{\sqrt{2^n n!}} \ket{n}.
\end{eqnarray}
Although coherent states and characteristic functions (see below) are usually introduced using the complex plane representation, we will continue using proper phase space vectors $\x$ since this vector representation simplifies the notation of Gaussian states (see below). 

For mixed quantum states in CV systems there is a useful representation by characteristic functions and quasi-probability distributions \cite{Adesso2014, Serafini2017}.The characteristic function of a mixed quantum state $\Rho$ is defined by
\begin{eqnarray}
\chi_{\Rho}^s (\x) = \Tr{\Rho \bm{D}(-\x)} \e^{\frac{s}{4} |\x|^2},
\end{eqnarray}
where $s \in [-1, 1]$. Its Fourier transform is the corresponding quasi-probability distribution 
\begin{eqnarray}
W_{\Rho}^s (\x) = \frac{1}{(2\pi)^2} \int_{\reals^2} d \vec{y} \  \chi_{\Rho}^s (\vec{y}) \e^{-i \vec{y}^T \Omegabm^T \x}.
\end{eqnarray}
Note that these distributions are real-valued, integrate to one but might be locally negative. For $s=0$ one finds the so called \textit{Wigner function} $W_{\Rho}^0 = W_{\Rho}$ \cite{Wigner1932}, in the case $s = -1$ we obtain the \textit{Husimi Q-function} $W_{\Rho}^{-1} = Q_{\Rho}$ \cite{Husimi1940}, while $s = 1$ yields the \textit{Glauber-Sudarshan P-function} $W_{\Rho}^1 = P_{\Rho}$ \cite{Glauber1963, Sudarshan1963}.

As already mentioned, coherent states also form a continuous quantum frame closely related to the P- and Q-functions. It can be shown \cite{Serafini2017} that
\begin{eqnarray}
\Rho &= \int_{\reals^2} d^2x \ P_{\Rho}(\x) \ketbra{\x}{\x}
\label{eq:P-function}
\end{eqnarray}
and, invoking the completeness relation $\1 = \frac{1}{\pi} \int_{\reals^2} d^2x \ketbra{\x}{\x}$, one also finds
\begin{eqnarray}
Q_{\Rho}(\x) \ &= \ \frac{1}{\pi} \braket{\x|\Rho|\x} = \Tr{\E(\x) \Rho},
\label{eq:Q-function}
\end{eqnarray}
where
\begin{equation}
 \E(\x) = \frac{1}{\pi} \ketbra{\x}{\x}
\end{equation}
are the elements of the continuous IC-POVM associated with the frame of coherent states. Thus, in the context of coherent states as quantum frame we find that the P-function and the Q-function are analogous to the frame vector $\vec{f}$ and to the IC-POVM probability distribution $\vec{p}$ of section \ref{Quantum Frames and IC-POVM}, respectively. This analogy also explains why the Q-function is a real probability distribution for any $\Rho$, i.e.~takes only non-negative values for any quantum state. Furthermore, we can use (\ref{eq:P-function}) and (\ref{eq:Q-function}) to express the transformation from one distribution to the other as
\begin{equation}
 Q(\x) = \mathcal{M}[P](\x) = \int_{\reals^2} d^2y \ m(\x, \vec{y}) P(\vec{y}),
\end{equation}
where
\begin{equation}
 m(\x, \vec{y}) = \frac{1}{\pi} |\braket{\x|\vec{y}}|^2 = \frac{1}{\pi} \e^{-|\x - \vec{y}|^2/2}.
\label{eq:P->Q}
\end{equation}
The transition kernel $m(\x, \vec{y})$ represents the continuous version of the matrix $M_{ij}$ in (\ref{eq:Mf=p}) and again basically consists of the kernel expression of the frame operator $s(\x, \vec{y}) = |\braket{\x|\vec{y}}|^2 = \e^{-|\x -\vec{y}|^2/2}$. And again one can derive the same chain of inequalities between the trace distance of two quantum states and the Kolmogorov distance between their Q- and P-functions
\begin{eqnarray}
\dis_\Kol(Q_1, Q_2) \ \leq \ \dis_\Trace(\Rho_1, \Rho_2) \ \leq \ \dis_\Kol(P_1, P_2).
\label{eq:Q-dis<tr-dis<P-dis}
\end{eqnarray}
Here, the Kolmogorov distance of two continuous distributions $R_1$ and $R_2$ over phase space is defined by
\begin{equation}
 \dis_\Kol(R_1, R_2) = \frac{1}{2} \int_{\reals^2} d^2x \ |R_1(\x) - R_2(\x)|.
\end{equation}
Again one can find the proof for both inequalities in appendix \ref{Proof of Eq. (eq:Q-dis<tr-dis<P-dis)}, though the second inequality can also been found in \cite{Nair2017}.

A special class of quantum states in CV systems are so called \textit{Gaussian states} which are usually defined by having a characteristic function for $s=0$ that takes a Gaussian form \cite{Adesso2014}
\begin{eqnarray}
\chi^0_{\Rho}(\x) \ = \ \Exp{-\frac{1}{4} \x^T \Omegabm^T \sigmabm \Omegabm \x \ + \ i \x^T \Omegabm^T \vec{d}},
\end{eqnarray}
where $\sigmabm$ captures the statistical moments of second order (covariances) and $\vec{d}$ the statistical moments of first order (mean values) of the operators $\q = \bm{X}_1$ and $\p = \bm{X}_2$
\begin{eqnarray}
\sigma_{ij} &= \braket{\bm{X}_i \bm{X}_j + \bm{X}_j \bm{X}_i} - 2\braket{\bm{X}_i}\braket{\bm{X}_j} \quad \und \quad d_i &= \braket{\bm{X}_i}.
\end{eqnarray}
With
\begin{eqnarray}
\chi_{\Rho}^s(\x) \ = \ \chi_{\Rho}^0(\x) \ \e^{\frac{s}{4}|\x|^2}
\end{eqnarray}
we see immediately that for any value of $s$ the characteristic function will be Gaussian as well and since Gaussianity is preserved under Fourier transformation same is true for any quasi-probability distribution. Using $|\vec{x}|^2 = \vec{x}^T \vec{x} = \vec{x}^T \Omegabm^T \Omegabm \vec{x}$ one finds
\begin{eqnarray}
\chi_{\Rho}^s(\x) \ &= \ \Exp{-\frac{1}{4} \x^T \Omegabm^T (\sigmabm - s\1) \Omegabm \x \ + \ i \x^T \Omegabm^T \vec{d}}
\\
&= \ \Exp{-\frac{1}{4} \x^T \Omegabm^T \bm{\sigma}_s \Omegabm \x \ + \ i \x^T \Omegabm^T \vec{d}}
\label{eq:Chi^s_Rho}
\end{eqnarray}
with $\bm{\sigma}_s := \sigmabm - s\1$ and by this
\begin{eqnarray}
W_{\Rho}^s(\x) \ = \ \frac{1}{\pi \sqrt{\det \bm{\sigma}_s}} \Exp{-(\x - \vec{d})^T \bm{\sigma}_s^{-1} (\x - \vec{d})}.
\end{eqnarray}
Note that for coherent states $\ketbra{\x}{\x}$ we have $\sigmabm = \1$ and $\vec{d} = \x$ such that the P function reduces to $P_{\ketbra{\x}{\x}}(\vec{y}) = W_{\ketbra{\x}{\x}}^{s=1}(\vec{y}) = \delta(\vec{y} - \x)$ as expected. Furtheron we define $\sigmabm_P := \bm{\sigma}_1$ and $\sigmabm_Q := \bm{\sigma}_{-1}$ as the covariance matrices for P and Q function over phase space. These definitions together with $\bm{\sigma}_s := \sigmabm - s\1$ lead to
\begin{eqnarray}
\sigmabm_Q = \sigmabm_P + 2\cdot\1,
\label{eq:sigma_Q=sigma_P+4}
\end{eqnarray}
i.e. for a given Gaussian quantum state its P- and Q-function are both Gaussian shaped, centered around the same mean value $\vec{d}$ and the Q-function is broadened compared to the P-function by an offset of $2\cdot\1$ in the covariance matrix.

\begin{figure}[h!]
\begin{subfigure}{0.49\columnwidth}
\includegraphics[width=\linewidth]{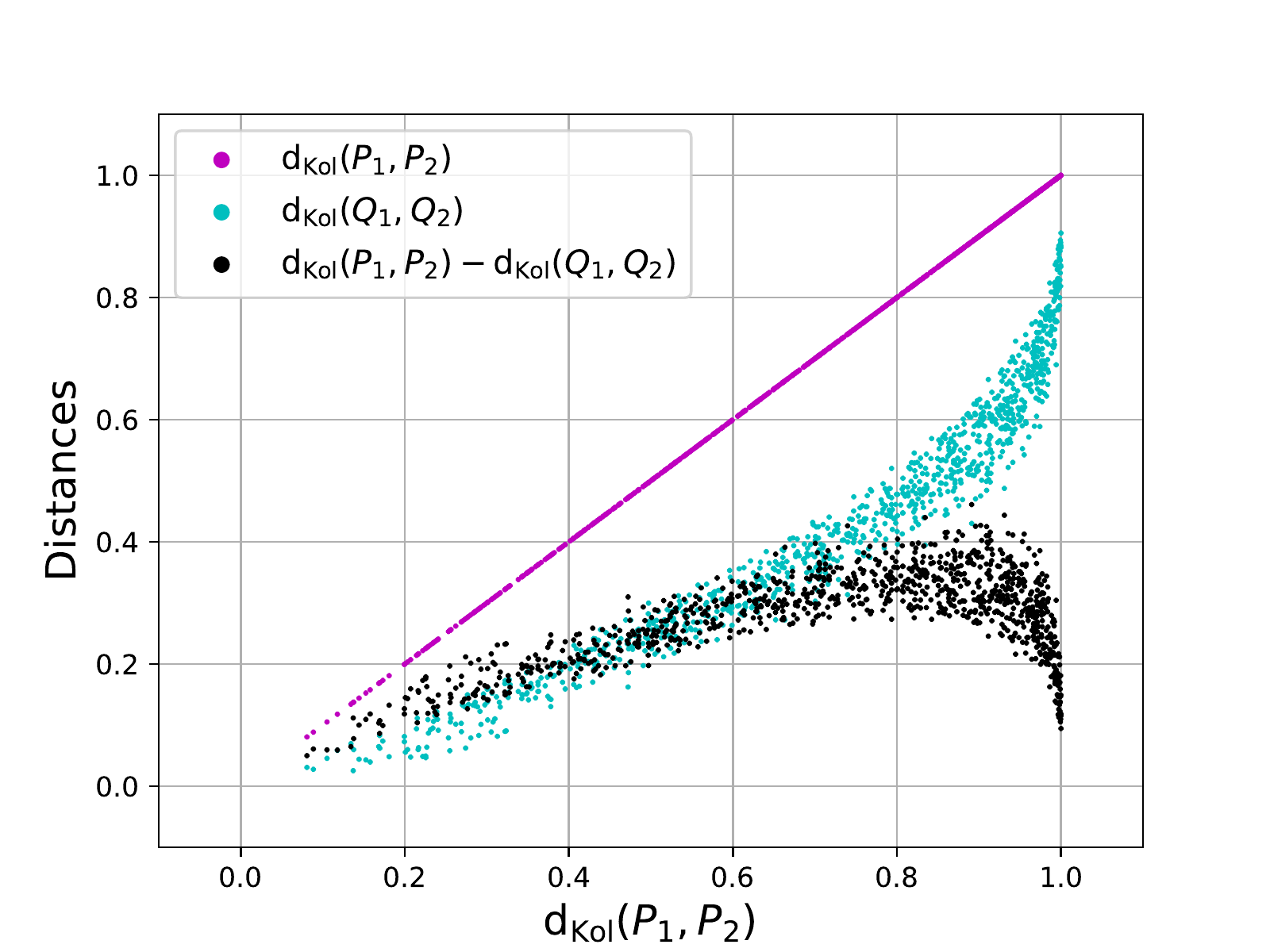}
\caption{$ \ \ d \in [0,2]$, $\varphi_d \in [-\pi, \pi]$ \\ $\qquad \sigma_{ii} \in [0.25, 1]$, $\varphi \in [0,2\pi]$}
\end{subfigure}\hfill
\begin{subfigure}{0.49\columnwidth}
\includegraphics[width=\linewidth]{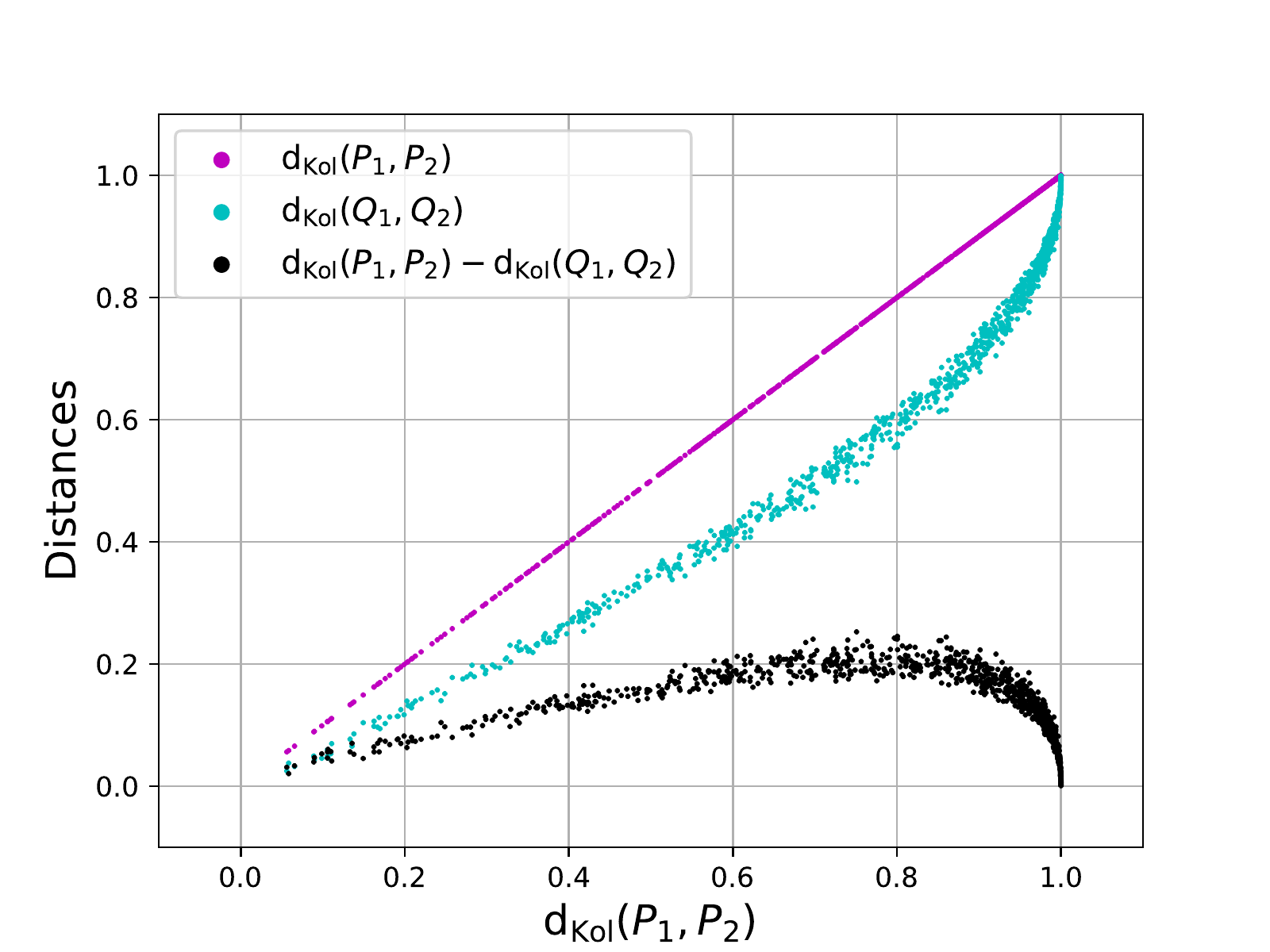}
\caption{$ \ \ d \in [0,5]$, $\varphi_d \in [-0.5\pi, 0.5\pi]$ \\ $\qquad \sigma_{ii} \in [1,2]$, $\varphi \in [0,2\pi]$}
\end{subfigure}\\[1em]
\begin{subfigure}{0.49\columnwidth}
\includegraphics[width=\linewidth]{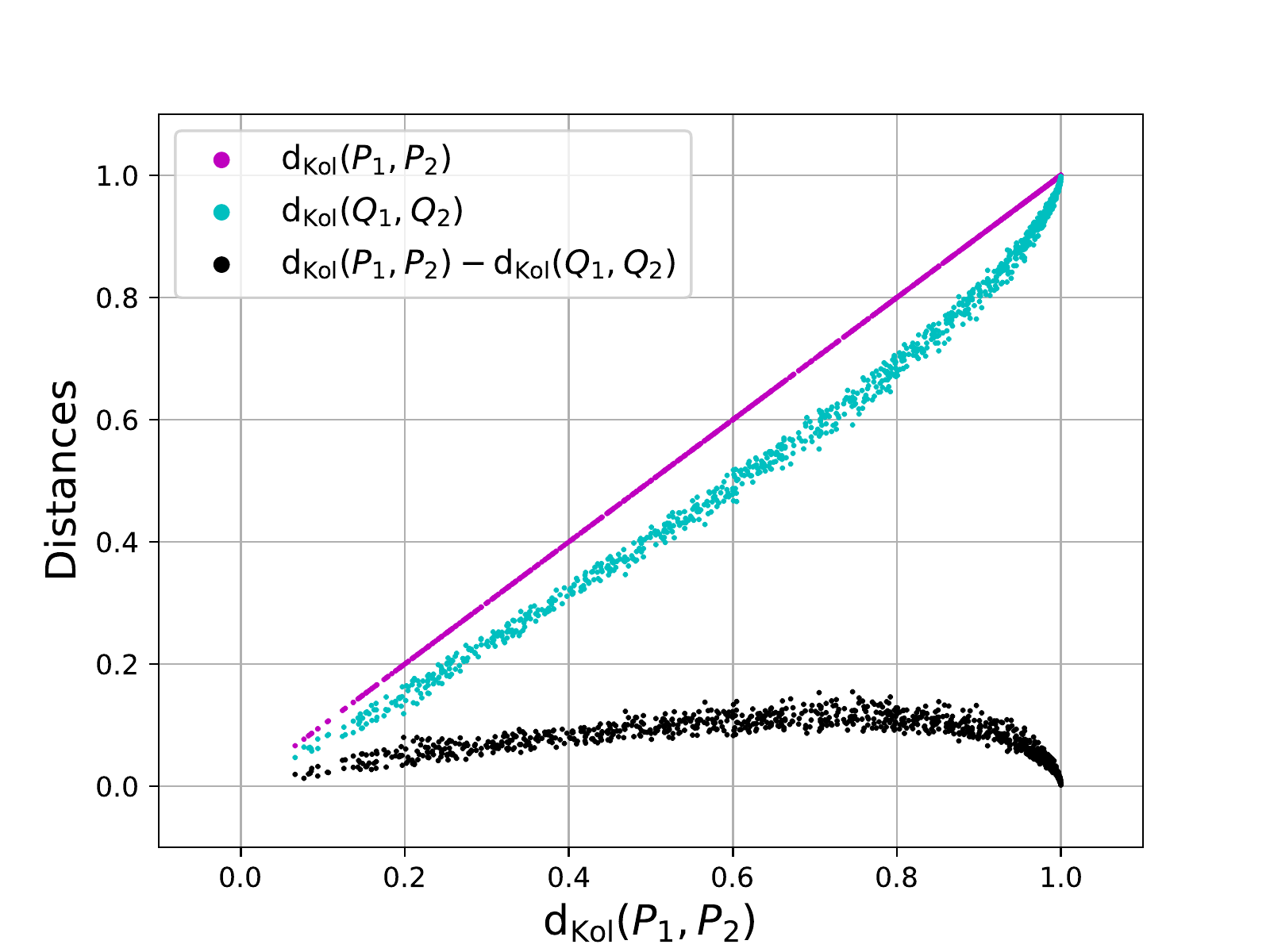}
\caption{$ \ \ d \in [0,10]$, $\varphi_d \in [-0.05\pi, 0.05\pi]$ \\ $\qquad \sigma_{ii} \in [2,5]$, $\varphi \in [0,2\pi]$}
\end{subfigure}\hfill
\begin{subfigure}{0.49\columnwidth}
\includegraphics[width=\linewidth]{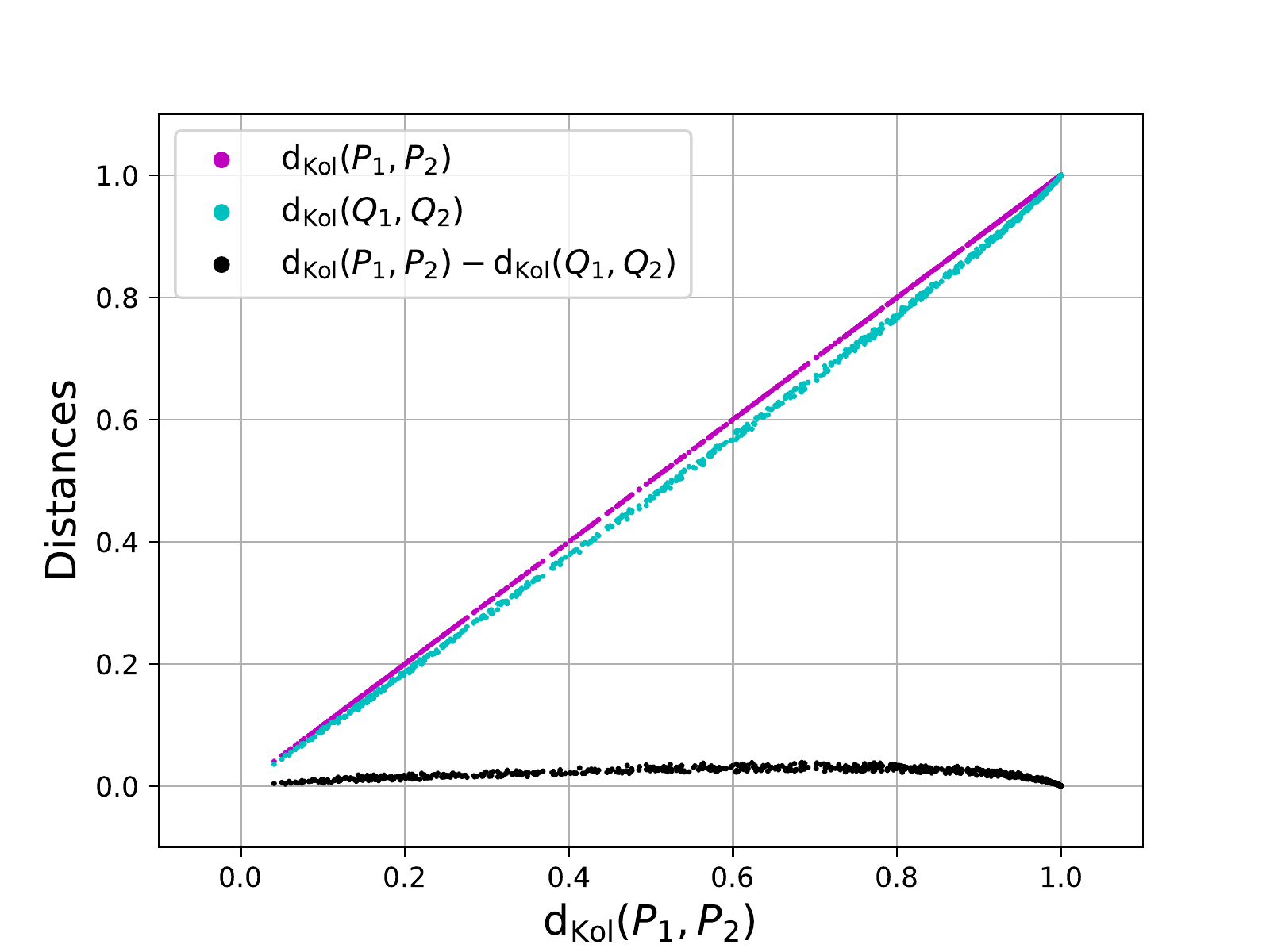}
\caption{$ \ \ d \in [0,30]$, $\varphi_d \in [-0.005\pi, 0.005\pi]$ \\ $\qquad \sigma_{ii} \in [10, 20]$, $\varphi \in [0,2\pi]$}
\end{subfigure}
\caption{\begin{footnotesize}A sample plot of $\dis_\Kol(P_1, P_2)$ and $\dis_\Kol(Q_1, Q_2)$ as well as their difference over $\dis_\Kol(P_1, P_2)$ for 1000 randomly drawn pairs of Gaussian quantum states in four different scenarios. According to (\ref{eq:Gaussian-state_parameters}) each Gaussian state is defined by its parameters $d$, $\varphi_d$, $\sigma_{00}$, $\sigma_{11}$ and $\varphi$ which are drawn uniformly from the intervals indicated in the captions to each subplot. The first scenario (a) describes a low entropy regime, while the last scenario (d) is characterized by large entropies, involving a much broader distributions of excited Fock states. The scenarios (b) and (c) illustrate the transition between both regimes. \end{footnotesize}}
\label{fig:Sample-Plots_coherent}
\end{figure}

Figure \ref{fig:Sample-Plots_coherent} shows a sample plot for a CV-system similar to Fig.~\ref{fig:Sample-Plot_Qubit-SIC.pdf} for the qubit. For simplicity we restrict ourselves to Gaussian quantum states as introduced above with their P-function given as
\begin{eqnarray}
P(\x) \ = \ \frac{1}{\pi\sqrt{\det \sigmabm_P}}\exp\left[-(\x-\vec{d})^T\sigmabm_P^{-1}(\x-\vec{d})\right].
\label{eq:Gaussian-P-function}
\end{eqnarray}
We plot the Kolmogorov distances $\dis_\Kol(P_1,P_2)$, $\dis_\Kol(Q_1,Q_2)$ and their difference $\dis_\Kol(P_1,P_2) - \dis_\Kol(Q_1,Q_2)$ over $\dis_\Kol(P_1,P_2)$ for different scenarios of state pairs. As stated above each Gaussian state is uniquely defined by the displacement of its center $\vec{d}$ and its covariance matrix $\sigmabm_P$. The displacement and the covariance matrix is expressed as
\begin{eqnarray}
\vec{d} = d \left(\begin{array}{cc} \cos(\varphi_d) \\ \sin(\varphi_d) \end{array}\right) \qquad \und \qquad \sigmabm_P = S(\varphi) \left(\begin{array}{cc} \sigma_{00} & 0 \\ 0 & \sigma_{11} \end{array}\right) S(\varphi)^T,
\label{eq:Gaussian-state_parameters}
\end{eqnarray}
where $S(\varphi)$ is a rotation by $\varphi$ in the phase plane. Thus we can characterize any Gaussian state by its randomly drawn parameters $d$, $\varphi_d$, $\sigma_{ii}$ and $\varphi$.

In Fig.~\ref{fig:Sample-Plots_coherent} we observe a transition of the differences between the two distances for increasing values of $\sigma_{ii}$. If on the one hand the $\sigma_{ii}$ are small the state $\Rho = \int_{\reals^2}d^2x P(\x) \ketbra{\x}{\x}$ has a low von Neumann entropy (see below) and we can approximate it by means of the matrix representation $\rho_{mn} = \braket{m|\Rho|n}$ using only a small number of Fock basis states $\ket{m}$. This is the scenario plotted in Fig.~\ref{fig:Sample-Plots_coherent}(a) from which one can immediately see that the difference between $\dis_\Kol(P_1,P_2)$ and $\dis_\Kol(Q_1,Q_2)$ is rather large and, hence, both distances do not yield a good approximation of the trace distance $\dis_\Trace$. This effect can be explained by the relation (\ref{eq:sigma_Q=sigma_P+4}) between the covariance matrices of the P-function and the Q-function. In case of small diagonal entries in $\sigmabm_P$ the offset of $2\cdot\1$ widens the distribution significantly. As a consequence, although two close but very thin P-functions might have a Kolmogorov distance of almost one, their corresponding Q-functions will show significant overlap and hence a much smaller Kolmogorov distance. However, note that in this low entropy regime a direct evaluation of the trace distance by means of a truncated Fock space representation is typically feasible and much more efficient.

On the other hand, if the $\sigma_{ii}$ are large this results in a large von Neumann entropy and due to the broad distribution a large number of relevant entries in the Fock-state representation. Figure \ref{fig:Sample-Plots_coherent}(d) shows this case and we see that $\dis_\Kol(P_1,P_2)$ and $\dis_\Kol(Q_1,Q_2)$ are always very close to each other. Again, we can understand this effect by using (\ref{eq:sigma_Q=sigma_P+4}) since now the width $\sigmabm_P$ is already so large that the offset of $2\cdot\1$ does not cause a significant change. Accordingly, in this large entropy regime the Kolmogorov distance approximates the trace distance quite well. This is of particular interest since the Fock-state matrix representation will include such a huge number of relevant entries that a direct computation of the trace distance would turn out to be highly challenging. The plots in Fig.~\ref{fig:Sample-Plots_coherent}(b) and (c) indicate the transition between the low and high entropy regimes, which might also be characterized as quantum and classical like regimes.

The connection of the different scenarios of Fig.~\ref{fig:Sample-Plots_coherent} to the von Neumann entropy is illustrated in Fig.~\ref{fig:Dis-dif_over_S}, where one finds for the state pairs of Fig.~\ref{fig:Sample-Plots_coherent} a plot of their differences of Kolmogorov distances between P- and Q-functions over their average von Neumann entropy $\frac{1}{2}\left(S_\mathrm{v.N.}(\Rho_1) + (S_\mathrm{v.N.}(\Rho_2) \right)$. The von Neumann entropy of a Gaussian state can be computed using the symplectic eigenvalues of its Wigner functions covariance matrix $\sigmabm_0 =: \sigmabm$ \cite{Serafini2017}. The symplectic eigenvalues $\pm v$ of $\sigmabm$ are the eigenvalues of the matrix $i \Omegabm \sigmabm$ which come as a pair of two real numbers differing only in their sign \cite{Serafini2017}, i.e. we can define \textit{the} symplectic eigenvalue $v$ as the positive one (in case of more than one mode we find one symplectic eigenvalue for each mode). The von Neumann entropy then reads \cite{Serafini2017}
\begin{equation}
 S_\mathrm{v.N.}(\Rho) = -\tr \{ \Rho \log_2 \Rho \} = \frac{v+1}{2} \log_2\left(\frac{v+1}{2} \right) - \frac{v-1}{2}
 \log_2\left(\frac{v-1}{2} \right).
\label{eq:S_v.N.}
\end{equation}
Of course the von Neumann entropy of two states in a pair might differ but since both states are generated using parameters from the same intervals, their von Neumann entropies are typically of the same order of magnitude. One finds a clear transition from large scattering in case of low entropies to very small scattering in case of high entropy. The four clusters seen in Fig.~\ref{fig:Dis-dif_over_S} correspond to the four scenarios shown in Fig.~\ref{fig:Sample-Plots_coherent}.

\begin{figure}[h!]
\center
\includegraphics[width=0.8\textwidth]{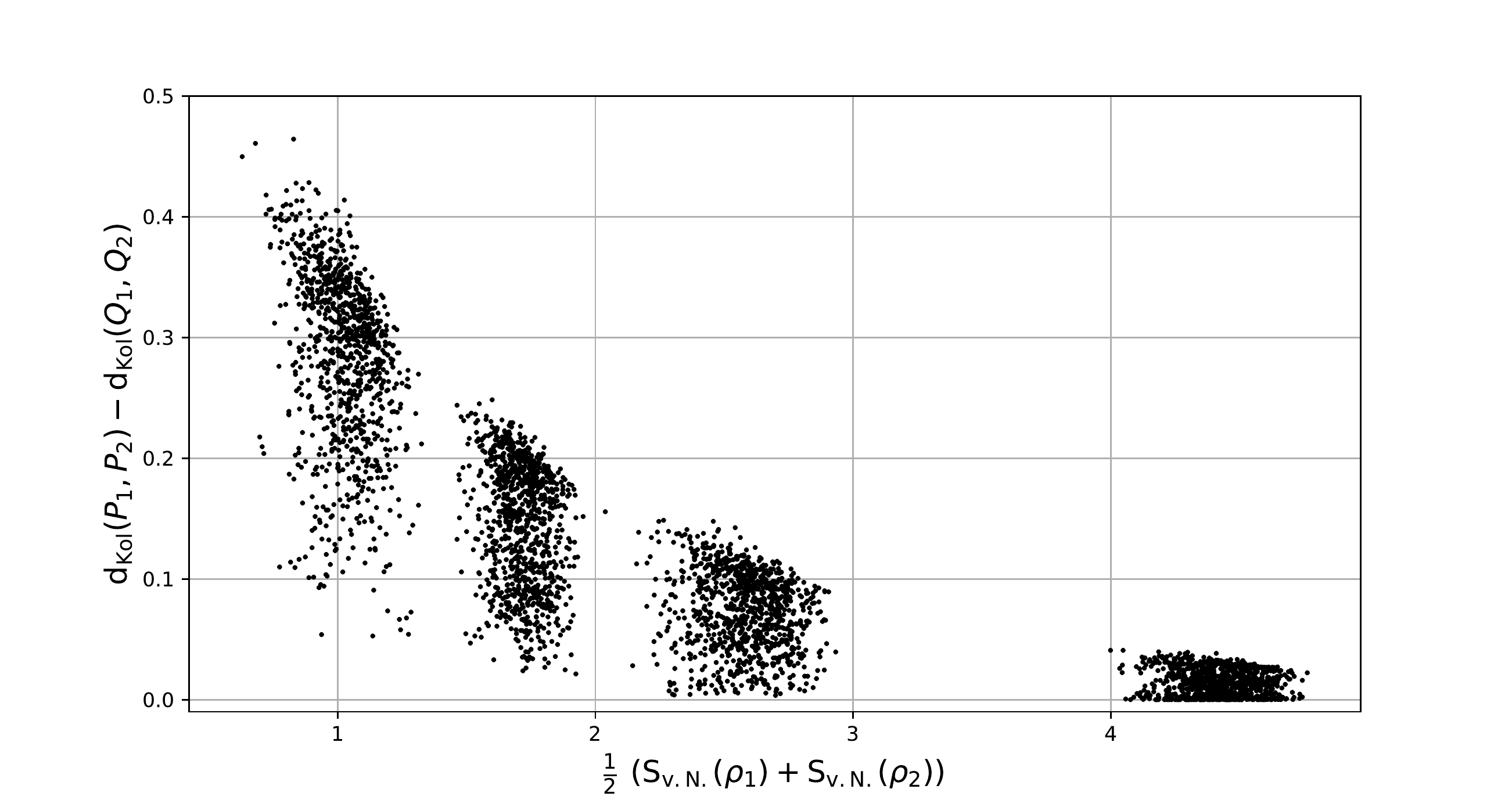}
\caption{\begin{footnotesize}Difference between the Kolmogorov distances between P- and Q-functions $\dis_\Kol(P_1,P_2) - \dis_\Kol(Q_1,Q_2)$ for each pair plotted over their average von-Neumann entropy $\frac{1}{2}\left(S_\mathrm{v.N.}(\Rho_1) + (S_\mathrm{v.N.}(\Rho_2) \right)$ computed numerically by using (\ref{eq:S_v.N.}).\end{footnotesize}}
\label{fig:Dis-dif_over_S}
\end{figure}

Finally, we have to make some general remarks concerning the plots of Fig.~\ref{fig:Sample-Plots_coherent}. Firstly, we see that even the Kolmogorov distance of the P functions - in contrast to the Kolmogorov distance of the frame vectors in the qubit case of Fig.\ref{fig:Sample-Plot_Qubit-SIC.pdf} - is bounded by one. This is simply due to the fact that the plot is restricted to Gaussian states for which even the P-function and the W-function are proper probability distributions. Secondly, we have to note that due to our restriction to Gaussian states we have to be a bit careful when drawing general conclusions about the behavior of Kolmogorov distances between quasi-probability distributions in an entropic transition. Nevertheless, Gaussian states are often understood as prototypical states of CV systems and are widely used in many fields \cite{Giedke2002, A.Ferraro2005, Adesso_PhD, Adesso2014, Link2015}. Thus, any knowledge of their behavior in an entropic transition is certainly of value and interest and might hint to more general statements of further investigations.

\section{Non-Markovianity and CV Systems}
\label{Non-Markovianity and CV-Systems}
\subsection{Witnessing non-Markovianity with P- and Q-functions}
\label{Witnessing non-Markovianity with Q and P-functions}
In this section we will use the distance measures introduced above in order to detect and quantify non-Markovianity, i.e. memory effects in the dynamics of open quantum systems described by a family of dynamical maps $\{\Lambda_t\}$. To this end, we employ the notion of quantum non-Markovianity based on the information flow between the open system and its environment. As explained in the introduction, this information flow is quantified by the trace distance between a pair of quantum states which represents a measure for the distinguishability of these states. However, as already mentioned in section \ref{Distance Measures on Quantum States} the calculation of the trace distance can become very demanding, in particular for CV systems. Though, based on the approximation of the trace distance by the Kolmogorov distances between P- and Q-functions as explained in section \ref{Coherent States as Continuous Quantum Frame}, we can define an easily accessible witness of non-Markovianity in such CV systems. In fact, if we find
\begin{equation}
\dis_\Kol \big(Q_1(t), Q_2(t)\big) > \dis_\Kol\big(P_1(s), P_2(s)\big) \quad {\mbox{for some}} \quad t > s 
\label{NM-witness}
\end{equation}
we conclude from the inequalities (\ref{eq:Q-dis<tr-dis<P-dis}) that the trace distance $\dis_\Trace\big(\Rho_1(\tau), \Rho_2(\tau)\big)$ increases for some $\tau \in [s,t]$, i.e.~non-Markovian effects emerge during this time interval. In a different context a similar strategy has been proposed in \cite{Brugger2022}.

\subsection{Example: Non-Markovian damped oscillator}
\label{Example: Non-Markovian damped oscillator}
We discuss the criterion (\ref{NM-witness}) with the help of the example of a damped harmonic oscillator given by the master equation
\begin{eqnarray} \label{MEQ}
\frac{d}{dt} \Rho(t) \ = -i \omega [\A^\dagger\A, \Rho] \ &+ \ \gamma_-(t)\left( \A\Rho(t)\A^\dagger - \frac{1}{2} \{\A^\dagger \A, \Rho(t) \} \right)\nonumber \\ &+ \ \gamma_+(t)\left( \A^\dagger\Rho(t)\A - \frac{1}{2} \{\A \A^\dagger, \Rho(t) \} \right)
\end{eqnarray}
with constant frequency $\omega$ and time dependent emission and absorption rates $\gamma_-(t)$ and $\gamma_+(t)$, respectively. If $\gamma_\pm(t) \geq 0$ for all $t\geq 0$ the dynamical map is CP divisible and, hence, describes a Markovian dynamics \cite{Breuer2016}. We therefore allow temporarily negative rates setting $\gamma_\pm(t) := \gamma_\pm \sin(\Omega t)$. Transforming the master equation to the interaction picture we thus have
\begin{eqnarray}
\frac{d}{dt} \Rho(t) \ &= \ \gamma_- \sin(\Omega t)\left( \A\Rho(t)\A^\dagger - \frac{1}{2} \{\A^\dagger \A, \Rho(t) \} \right)\nonumber \\ &+ \ \gamma_+ \sin(\Omega t)\left( \A^\dagger\Rho(t)\A - \frac{1}{2} \{\A \A^\dagger, \Rho(t) \} \right).
\label{eq:driven_dumping}
\end{eqnarray}
In the following we assume $\gamma_- \geq \gamma_+$ and set $\gamma_0 := \gamma_- - \gamma_+ \geq 0$. As is shown in appendix \ref{Equations of Motion for Driven Dumping} the master equation (\ref{eq:driven_dumping}) can easily be solved for Gaussian initial states by determining the mean values $\braket{\A}(t)$, $\braket{\A^2}(t)$ and $\braket{\A^\dagger \A}(t)$ which directly yield the time evolution of $\vec{d}(t)$ and $\sigmabm_P(t)$:
\begin{eqnarray}
\vec{d}(t) \ &= \ \vec{d}(0) D(t) \\
\sigmabm_P(t) \ &= \ \sigmabm_P(0) D(t)^2 + \frac{2\gamma_+}{\gamma_0}\big(1-D(t)^2\big) \1 ,
\end{eqnarray}
where we have introduced
\begin{equation} \label{D-t}
D(t) := \e^{-\frac{\gamma_0}{2\Omega}\big( 1 - \cos(\Omega t) \big)}.
\end{equation}
The condition $\gamma_0 \geq 0$ ensures that $D(t)\leq 1$, and we see that $\gamma_\pm/\Omega$ can be understood as a measure for the size of the coupling (for emission and absorption, respectively).

Figure \ref{fig:Witnessing_Non-Markovianity} shows an example of the time evolution of the Kolmogorov distances between P- and Q-functions for the non-Markovian damped oscillator. We can immediately see that the criterion (\ref{NM-witness}) can be applied here to prove non-Markovianity of the dynamics. It can even be used to obtain a lower bound $\mathcal{N}_\mathrm{min}$ for the trace distance based non-Markovianity measure defined by (\ref{NM-measure}). In order to approximate this measure for a time interval $[s,t]$ between a minimum of $\dis_\Kol(P_1,P_2)$ and a following maximum of $\dis_\Kol(Q_1,Q_2)$ we use inequalities (\ref{NM-witness}) to get
\begin{eqnarray}
\mathcal{N}(\Rho_1, \Rho_2, \{\Lambda_\tau| \tau \in [s,t]\})
&\geq& \dis_\Trace\big( \Rho_1(s), \Rho_2(s) \big) - \dis_\Trace\big( \Rho_1(t), \Rho_2(t) \big) \nonumber \\
&\geq& \dis_\Kol\big(Q_1(t),Q_2(t)\big) - \dis_\Kol\big(P_1(s),P_2(s)\big) \nonumber \\
&=:& \mathcal{N}_\mathrm{min}(\Rho_1, \Rho_2, \{\Lambda_\tau| \tau \in [s,t]\}).
\end{eqnarray}
For example, in the case shown in Fig.~\ref{fig:Witnessing_Non-Markovianity} we can use $[s,t] = [0.5 , 1.0]$ which gives $\mathcal{N}(\Rho_1, \Rho_2, \{\Lambda_\tau| \tau \in [0.5,1.0]\}) \ \geq \ 0.876$.

\begin{figure}[htbp]
\center
\includegraphics[width=0.7\textwidth]{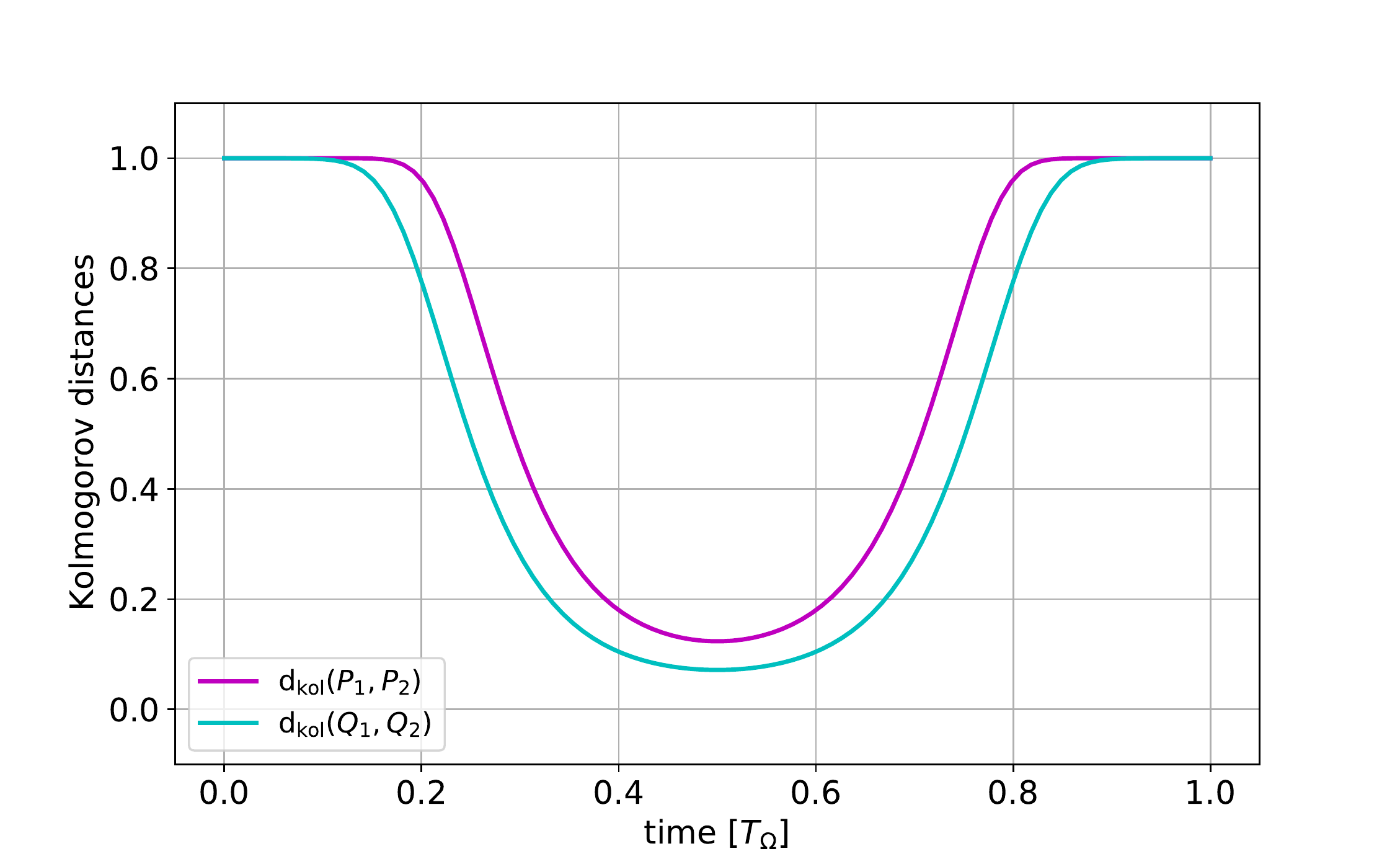}
\caption{\begin{footnotesize}Plot of $\dis_\Kol(P_1(t), P_2(t))$ and $\dis_\Kol(Q_1(t), Q_2(t))$ for the master equation (\ref{eq:driven_dumping}) with $\Omega = 2\pi$, $\gamma_- = 10\pi$ and $\gamma_+ = 2\pi$. The initial states are both Gaussian with $\vec{d}_{1,2} = \pm6$ and $\sigmabm_{1,2} = 2\1$ as displacement and covariance matrix of their P-functions.\end{footnotesize}}
\label{fig:Witnessing_Non-Markovianity}
\end{figure}

\begin{figure}[htbp]
\center
\includegraphics[width=0.8\textwidth]{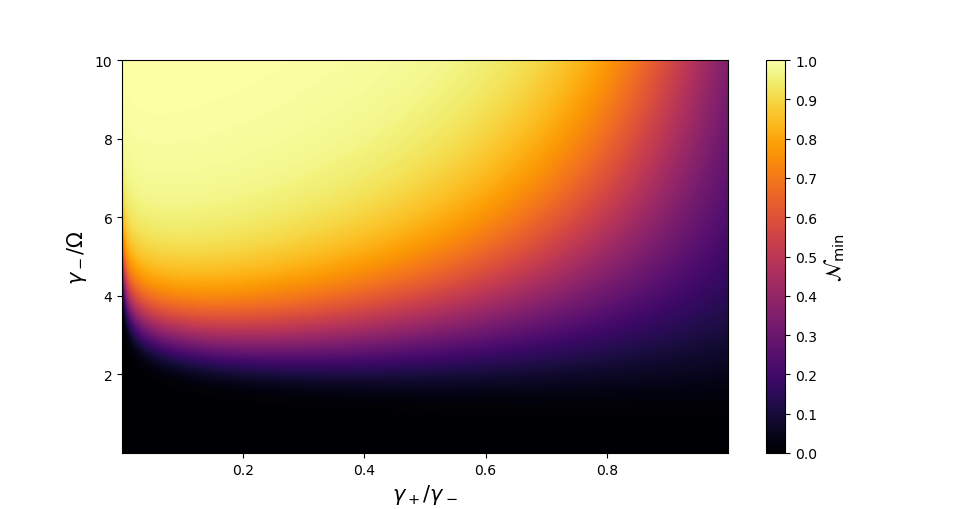}
\caption{\begin{footnotesize}Color plot of the lower bound of the non-Markovianity $\mathcal{N}_\mathrm{min}$ for the dynamics given by the master equation (\ref{eq:driven_dumping}) over varying parameters $\frac{\gamma_+}{\gamma_-}$ and $\frac{\gamma_-}{\Omega}$. The initial states are the same as those used in Fig.~\ref{fig:Witnessing_Non-Markovianity}.
\end{footnotesize}}
\label{fig:non-Markovianity-variation}
\end{figure}

\begin{figure}[h!]
\begin{subfigure}{0.49\columnwidth}
\includegraphics[width=\linewidth]{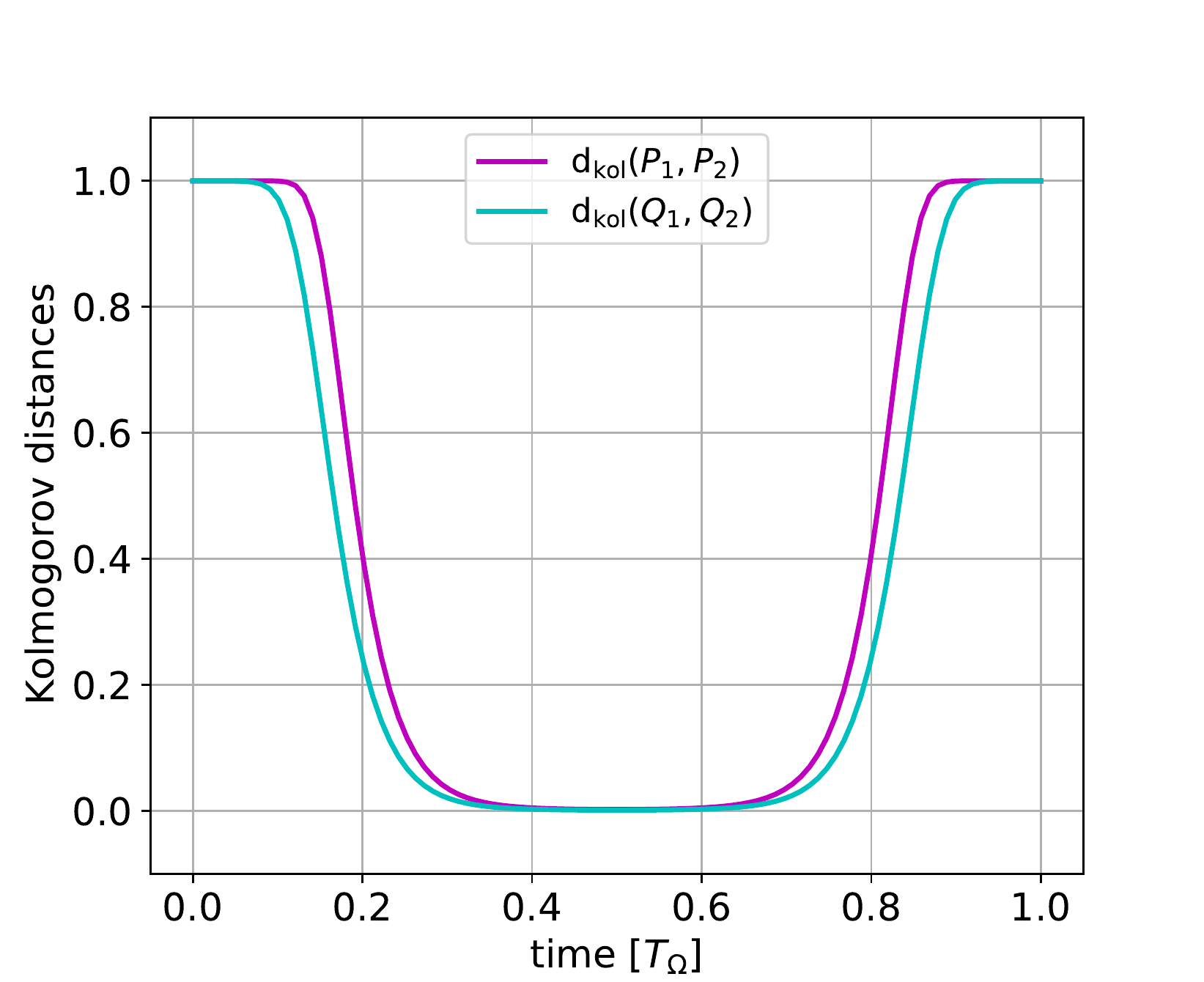}
\caption{$\frac{\gamma_+}{\gamma_-} = 0.2$ and $\frac{\gamma_-}{\Omega}=10$}
\end{subfigure}\hfill
\begin{subfigure}{0.49\columnwidth}
\includegraphics[width=\linewidth]{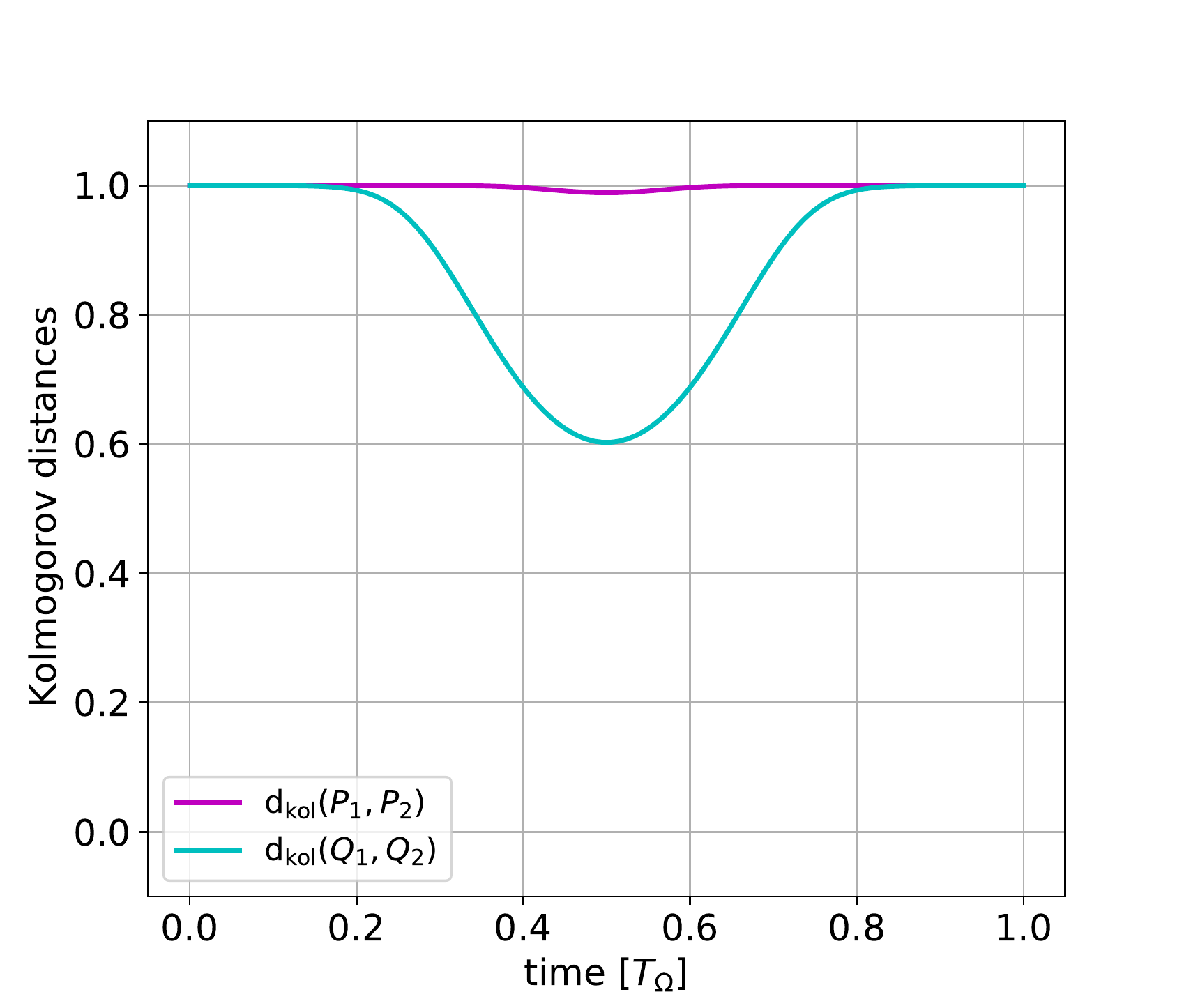}
\caption{$\frac{\gamma_+}{\gamma_-} = 0.05$ and $\frac{\gamma_-}{\Omega}=2$}
\end{subfigure}\\[1em]
\begin{subfigure}{0.49\columnwidth}
\includegraphics[width=\linewidth]{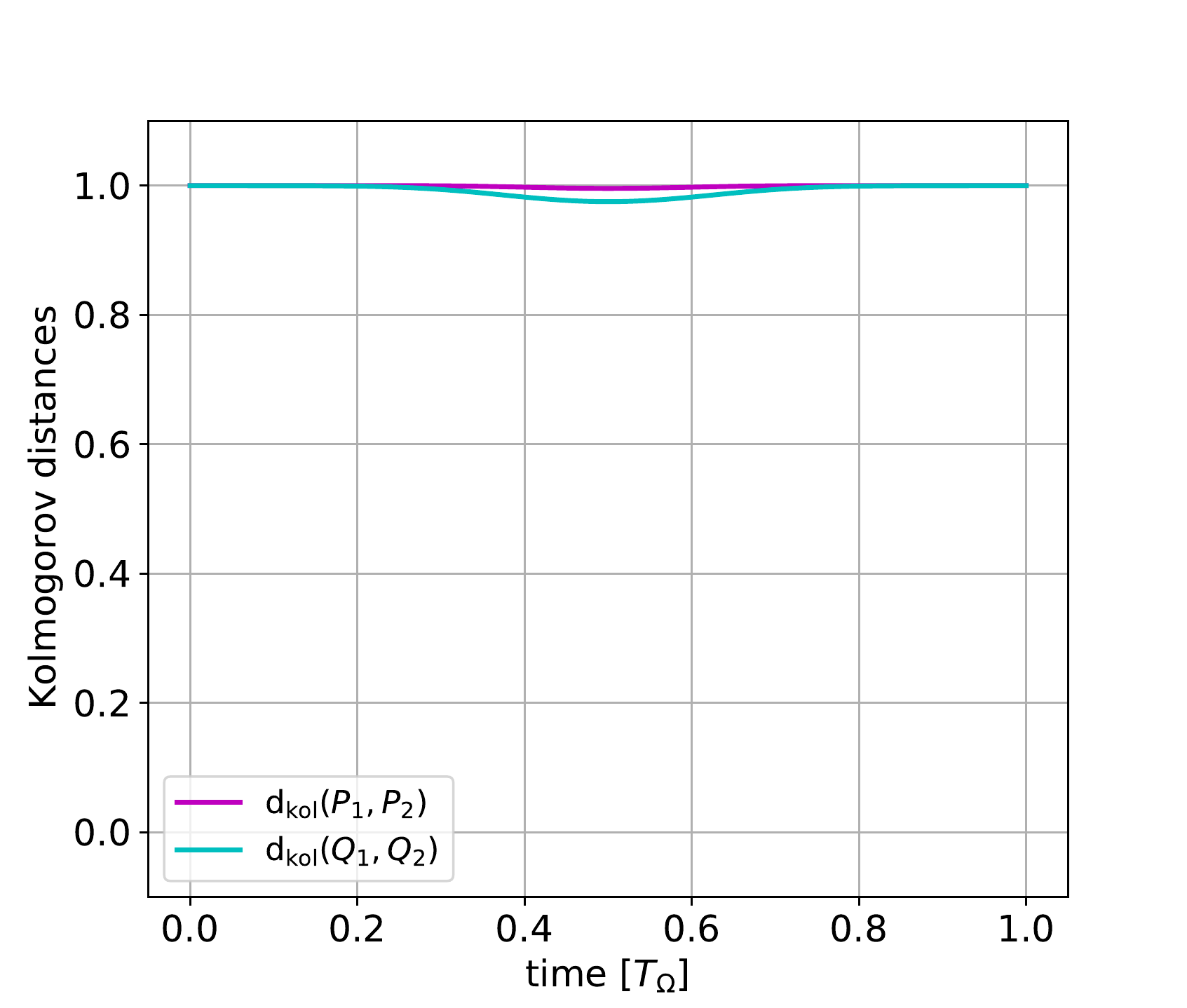}
\caption{$\frac{\gamma_+}{\gamma_-} = 0.5$ and $\frac{\gamma_-}{\Omega}=1$}
\end{subfigure}\hfill
\begin{subfigure}{0.49\columnwidth}
\includegraphics[width=\linewidth]{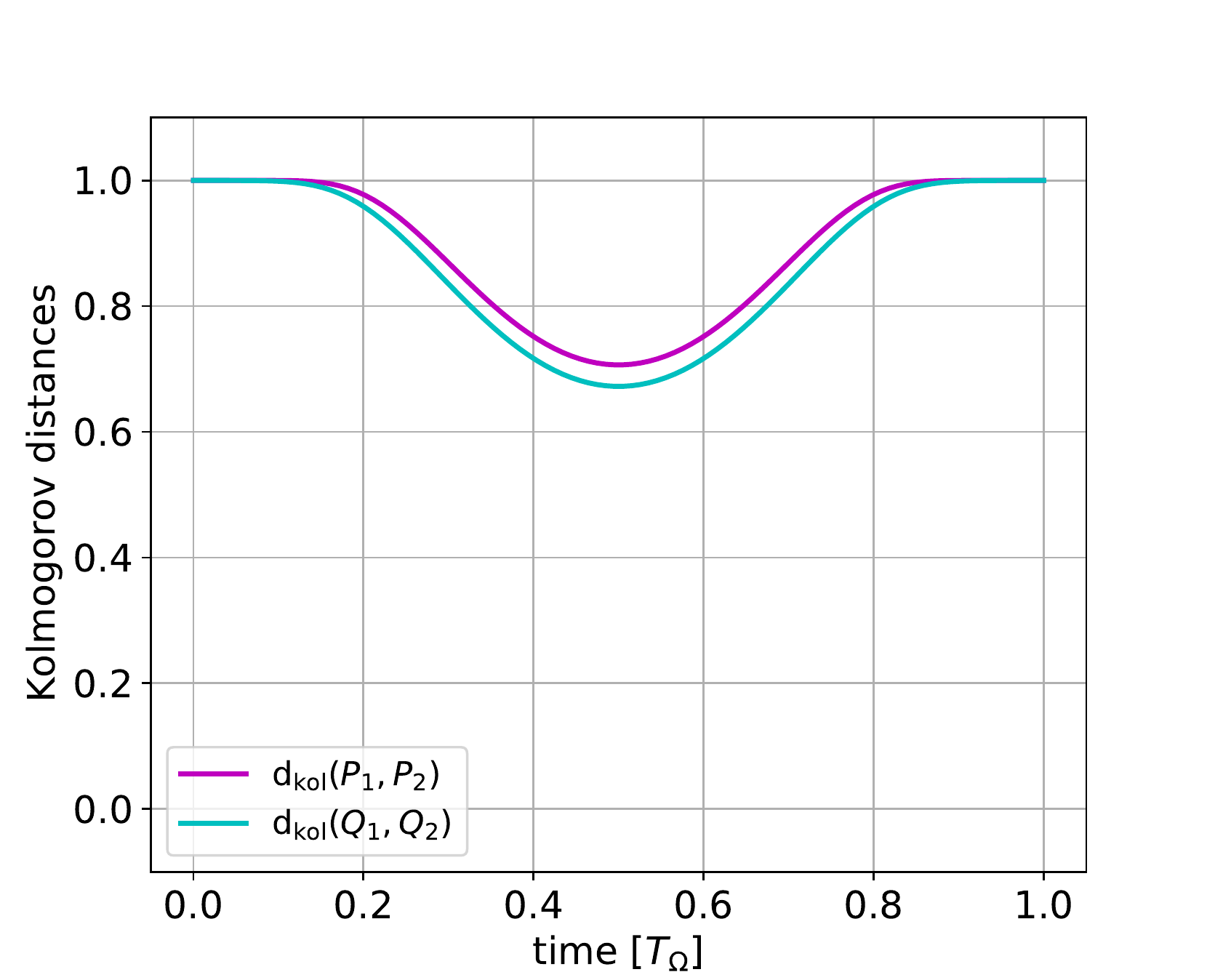}
\caption{$\frac{\gamma_+}{\gamma_-} = 0.8$ and $\frac{\gamma_-}{\Omega}=4$}
\end{subfigure}
\caption{\begin{footnotesize}Plot of the evolution of Kolmogorov distances for different parameter values for the dynamics given by (\ref{eq:driven_dumping}). For all plots the initial states are chosen according to Fig.~\ref{fig:Witnessing_Non-Markovianity} and $\Omega = 2 \pi$ always. Thus, subplot (a) can be found in the upper left corner, subplot (b) in the lower left corner, subplot (c) in the bottom middle and subplot (d) in the right middle part of Fig.~\ref{fig:non-Markovianity-variation} respectively.
\end{footnotesize}}
\label{fig:different_Dynamic-scenarios}
\end{figure}

Figure \ref{fig:non-Markovianity-variation} shows the lower bound $\mathcal{N}_\mathrm{min}$ of the non-Markovianity as a function of the parameters  $\frac{\gamma_+}{\gamma_-}$ and $\frac{\gamma_-}{\Omega}$ for the same initial states as used in Fig.~\ref{fig:Witnessing_Non-Markovianity}. As one can see for $\frac{\gamma_+}{\gamma_-} \to 1$ as well as for $\frac{\gamma_-}{\Omega} \to 0$ the witnessed non-Markovianity $\mathcal{N}_\mathrm{min}$ approaches zero. This can be understood if we rewrite expression (\ref{D-t}) for $D(t)$, which mainly controls the dynamics and the non-Markovianity of the process, as follows
\begin{equation}
D(t) = \exp \left[ -\frac{1}{2} \frac{\gamma_-}{\Omega} \left(1-\frac{\gamma_+}{\gamma_-}\right) \Big(1 - \cos(\Omega t)\Big) \right].
\end{equation}
In Fig.~\ref{fig:different_Dynamic-scenarios} one can find the evolution of the two Kolmogorov distances for different parameters reflecting four different areas in the plot of Fig.~\ref{fig:non-Markovianity-variation}. Subfigure (b) and (c) are both cases in which we can witness almost no non-Markovianity, although for completely different reasons. In the case of (b) the approximation of the trace distance by the two Kolmogorov distances is for the whole time interval rather inaccurate since both Kolmogorov distances differ significantly. Hence, we just cannot \textit{witness} a lot of non-Markovianity although there might be much more. Yet, in the case of (c) the approximation is sufficiently accurate but the dynamics shows just a small amount of non-Markovianity due to the weak coupling to the environment. The maximal non-Markovianity can be witnessed in subplot (a) reflecting the ideal parameter range for witnessing non-Markovianity, while in subplot (d) the witness works sufficiently well, too, yet the dynamics only leads to a medium amount of non-Markovianity.

\subsection{Gaussian Evolution of P-Function Kolmogorov Distance}
\label{Gaussian Evolution of P-Function Kolmogorov Distance}

Gaussian states are of special interest e.g. in quantum optics where they are often seen as a typical class of almost classical states (since for all values of $s$ one finds an actually positive function $W_{\Rho}^s$) as well as in CV quantum information since they are easily experimentally realizable and controllable \cite{Weedbrook2012, Adesso2014}. Accordingly, one often focuses on CPTP dynamical maps which preserve the Gaussian character of the input state and are hence adressed as \textit{Gaussian quantum channels} \cite{Adesso2014, Serafini2017, Grimmer2018}.

For any CPTP dynamical map $\Lambda$ on a CV system we can define its action on P-functions as
\begin{equation}
 \Lambda[P_{\Rho}](\x) =  P_{\Lambda(\Rho)}(\x),
\end{equation}
which yields the explicit representation
\begin{equation}
 P_{\Lambda(\Rho)}(\x) = \int_{\reals^2} d^2y \ \lambda(\x, \vec{y})P_{\Rho}(\vec{y}),
\end{equation}
where 
\begin{equation}
 \lambda(\x,\vec{y}) = P_{\Lambda(\ketbra{\vec{y}}{\vec{y}})}(\x)
\end{equation}
is the P-function of the image under $\Lambda$ of the coherent state $\ketbra{\vec{y}}{\vec{y}}$. The kernel $\lambda$ can be understood as continuous analog of the transition matrix mapping an input distribution to an output distribution for a discrete classical stochastic process \cite{Asmussen2003}. This leads us to the following theorem.

\begin{thm}\label{thm:Gaussian-contraction}
For any pair of quantum states $\Rho_1$ and $\Rho_2$ on a CV quantum system the Kolmogorov distance between their Glauber-Sudarshan P-functions $P_1$ and $P_2$ is contracting under any Gaussian quantum channel.
\end{thm}

\begin{proof}
Let $P' = \Lambda[P]$ be the image of $P$ under a Gaussian quantum channel $\Lambda$. We find
\begin{eqnarray}
\dis_\mathrm{kol}(P'_1, P'_2) &=& \frac{1}{2} \int_{\reals^2} d^2x \left| \int_{\reals^2} d^2y \  \lambda(\x, \vec{y})  (P_1(\vec{y}) - P_2(\vec{y})) \right|
\\
&\leq& \ \frac{1}{2} \int_{\reals^2} d^2x \int_{\reals^2} d^2y \  |\lambda(\x, \vec{y})| \cdot |P_1(\vec{y}) - P_2(\vec{y})|
\\
&=& \ \frac{1}{2} \int_{\reals^2} d^2x \int_{\reals^2} d^2y \  \lambda(\x, \vec{y})  |P_1(\vec{y}) - P_2(\vec{y})|
\\
&=& \ \frac{1}{2} \int_{\reals^2} d^2y \ \ubs{= 1}{\left(\int_{\reals^2} d^2x \  P_{\Lambda(\ketbra{\vec{y}}{\vec{y}})}(\x)\right)}  |P_1(\vec{y}) - P_2(\vec{y})|
\\
&=& \ \dis_\mathrm{kol}(P_1, P_2),
\end{eqnarray}
where we have used the fact that coherent states are Gaussian and, hence, $\lambda(\x,\vec{y})=P_{\Lambda(\ketbra{\vec{y}}{\vec{y}})}(\x)$ has to be again a Gaussian state for a Gaussian channel $\Lambda$, and is therefore a positive function. 
\end{proof}

This theorem has several interesting consequences. First, it implies that any increase of the Kolmogorov distance between two P-functions is due to a violation of the CPTP property or to a violation of Gaussianity preservation. Let us first suppose we have a family of dynamical maps $\{\Lambda_t\}$ which is known to preserve Gaussianity for all $t$. We can then use the Kolmogorov distance of the P-functions instead of the trace distance in the definition of non-Markovianity \cite{Breuer2009b, Breuer2016} which leads to the novel non-Markvianity measure
\begin{eqnarray}
\mathcal{N}_P(\Rho_1, \Rho_2, \{\Lambda_t\}) := \underset{\delta \geq 0}{\int} dt \ \delta(\Rho_1, \Rho_2, t)
\\
\with \quad \delta(\Rho_1, \Rho_2, t) := \frac{d}{dt} \dis_\Kol\big(\Lambda_t[P_1], \Lambda_t[P_2] \big).
\end{eqnarray}
This quantity measures the backflow of information for a certain pair of quantum states by means of the Kolmogorov distance of the corresponding P-functions $P_1$ and $P_2$. By maximizing over all possible pairs of initial states we obtain
\begin{eqnarray}
\mathcal{N}_P(\{ \Lambda_t \}) \ := \ \underset{\Rho_1, \Rho_2}{\mathrm{max}} \ 
\mathcal{N}_P(\Rho_1, \Rho_2, \{ \Lambda_t \}).
\end{eqnarray}
This is a measure of the non-Markovianity of a given family of Gaussian dynamical maps $\{ \Lambda_t \}$ which is analogous to the trace distance based measure defined by Eqs.~(\ref{NM-measure})-(\ref{NM-measure-max}), but which is significantly easier to use in the case of CV systems. An example is given by the non-Markovian damped oscillator discussed in the previous section \ref{Example: Non-Markovian damped oscillator} which is obviously described by a Gaussianity preserving family of dynamical maps. We emphasize again that this measure only requires the dynamical maps $\Lambda_t$ to be Gaussianity preserving, while the initial states are allowed to be arbitrary although we chose them to be Gaussian in Fig.~\ref{fig:Witnessing_Non-Markovianity} for reasons of computational simplicity. A further example is the Caldeira-Leggett model of quantum Brownian motion \cite{Caldeira1983,Grabert1988}, which may be treated analogously to the procedure used in \cite{Einsiedler2020}.

Finally we suggest another application of theorem \ref{thm:Gaussian-contraction}. If we know for sure that the dynamical map is CP divisible (see, e.g.~\cite{Breuer2016}) any increase of the Kolmogorov distance between the P-functions provides a witness for the non-Gaussianity of the dynamical map. Note that this criterion of the increase of the Kolmogorov distance only represents a sufficient condition for non-Gaussianity since we might think of a CPTP map which maps a coherent state to a non-Gaussian state with nevertheless positive P function. Since the proof relies only on the positivity of $\lambda(\vec{y}, \x)$ and not on its Gaussianity, this would result in a contraction of the Kolmogorov distance as well.

\section{Conclusion}
\label{Conclusion}
In this paper we used a vector like representations of quantum states by (quasi) probability distributions based on quantum frames and their induced IC-POVMs. The first we called the frame vector $\vec{f}$ and the latter the IC-POVM probability vector $\vec{p}$ of a quantum state $\Rho$. To distinguish two quantum states $\Rho_1$ and $\Rho_2$ the Kolmogorov distance (as a distance measure for probability distributions) was applied to both $\vec{f}_{1,2}$ and $\vec{p}_{1,2}$ to find an upper bound (by $\dis_\Kol(\vec{f}_1, \vec{f}_2)$) and a lower bound (by $\dis_\Kol(\vec{p}_1, \vec{p}_2)$) of the trace distance between $\Rho_1$ and $\Rho_2$.

Having discussed these bounds for the case of a SIC-POVM and its symmetric quantum frame on a qubit, we applied this idea to the continuous quantum frame given by the coherent states in a continuous variable (CV) system. We saw that the frame vector $\vec{f}$ has to be replaced by the Glauber-Sudarshan P-function and the probability vector $\vec{p}$ by the Husimi Q-function - both two widely used quasi-probability distributions for quantum states. As for the discrete case an inequality chain holds for these continuous quasi-probability distributions, providing again upper and lower bounds for the trace distance. We illustrated the performance of these bounds for a class of randomly generated Gaussian states and demonstrated that the bounds become steadily tighter for high entropy states. Thus, for large-entropy quantum states one can approximate successfully their trace distance by the Kolmogorov distances between their P- and Q-functions. This is of particular interest since the numerical computation of the trace distance of quantum states becomes highly demanding in the regime of large entropies, while on the other hand an efficient determination of the trace distance is relatively easy for low entropy states in a suitable basis of the state space. Thus, the method proposed here is in a certain sense complementary to the direct determination of the trace distance by diagonalization in low-dimensional Hilbert spaces.

Building on these result we constructed a witness for the non-Markovianity of the dynamics of open quantum systems which is based on the Kolmogorov distances between P- and Q-functions and applied it to the model of a non-Markovian oscillator. The witness was also shown to lead to a lower bound for the non-Markovianity measure based on the trace distance which is again particularly efficient in the regime of high entropy states. Furthermore, we showed that the Kolmogorov distance between P-functions has to contract under any Gaussian CPTP map (Gaussian quantum channel). Based on this we derived a proper non-Markovianity measure for Gaussian quantum dynamical maps, i.e. maps that preserve the Gaussianity of quantum states. Finally, we suggested another application of the contraction property, namely to use it as a witness for non-Gaussianity in case one knows for sure the dynamical map to be Markovian.

For future work we suggest to apply the non-Markovianity measure for Gaussian dynamics to more complex systems, e.g.~the Caldeira-Leggett model of quantum Brownian motion, or to the dynamics of open systems described by more realistic non-Markovian quantum master equations. Beyond the regime of Gaussian dynamics one can of course still use our non-Markovianity witness, which might be a possible approach to investigate non-Markovian interacting many-body systems. To this end, it will be helpful to find a more detailed characterization of the scenarios in which this witness works especially efficient and to investigate if the observed quality of the approximation for high entropy quantum states actually holds in more general situations.

\ack
We thank Jonathan Brugger, Christoph Dittel, Andreas Buchleitner and Tanja Schilling for fruitful discussions. This work has been supported by the German Research Foundation (DFG) through FOR 5099.

\appendix
\section*{Appendix}
\setcounter{section}{0}
\renewcommand{\thesection}{\Alph{section}}

\section{Proof of Eq. (\ref{eq:p-dis<tr-dis<f-dis})}
\label{Proof of Eq. (eq:p-dis<tr-dis<f-dis)}
We start with the first inequality which states that the Kolmogorov distance between two IC-POVM probability vectors is smaller or equal to the trace distance between the quantum states at hand. Introducing the spectral decomposition of $\Deltabm := \Rho_1 - \Rho_2 =  \sum_a \Delta_a \ketbra{a}{a}$ we obtain
\begin{eqnarray}
\dis_\mathrm{kol.}(\vec{p}_1, \vec{p}_2) &= \frac{1}{2} \sum_i |p_i^1 - p_i^2| = \frac{1}{2} \sum_i |\Tr{\E_i \Rho_1} - \Tr{\E_i \Rho_2}|
\\
&= \frac{1}{2} \sum_i |\Tr{\E_i \Deltabm}| = \frac{1}{2} \sum_i |\Tr{\sum_a \Delta_a \E_i \ketbra{a}{a}}|
\\
&\leq \ \frac{1}{2} \sum_{i, a} |\Delta_a \ubs{\geq 0}{\Tr{\E_i \ketbra{a}{a}}}| = \frac{1}{2} \sum_{i, a} |\Delta_a| \ \Tr{\E_i \ketbra{a}{a}}
\\
&= \frac{1}{2} \sum_a |\Delta_a| \ \ubs{= \Tr{\1 \ketbra{a}{a}} = 1}{\Tr{\sum_i \E_i \ketbra{a}{a}}} = \frac{1}{2} \ubs{= \Trace |\Deltabm|}{\sum_a |\Delta_a|}
\\
&= \frac{1}{2} \Trace|\Rho_1 - \Rho_2| = \dis_\Trace(\Rho_1, \Rho_2).
\end{eqnarray}
Note that we actually do not need $\{ \E_i\}$ to be an IC-POVM associated to a quantum frame. In fact, the first inequality is actually true for any POVM.

For the second inequality stating that the trace distance between two quantum states is smaller or equal to the Kolmogorov distance between their quantum frame vectors we use $\Rho_1 - \Rho_2 = \sum_i (f_i^1 - f_i^2) \ketbra{\psi_i}{\psi_i}$ to get
\begin{eqnarray}
&\dis_\Trace (\Rho_1, \Rho_2) = \frac{1}{2} ||\Rho_1 - \Rho_2|| 
= \frac{1}{2} || \sum_i (f_i^1 - f_i^2) \ketbra{\psi_i}{\psi_i} ||
\\
&\leq \frac{1}{2} \sum_i || (f_i^1 - f_i^2) \ketbra{\psi_i}{\psi_i} || 
= \frac{1}{2} \sum_i |f_i^1 - f_i^2| \cdot ||\ketbra{\psi_i}{\psi_i}||
\\
&= \frac{1}{2} \sum_i |f_i^1 - f_i^2| = \dis_\Kol(\vec{f}_1, \vec{f}_2),
\end{eqnarray}
where we have used that the trace norm satisfies the triangular inequality and that the trace norm of a rank-one projection is equal to one. \hfill $\quad \square$

\section{Proof of Eq. (\ref{eq:Q-dis<tr-dis<P-dis})}
\label{Proof of Eq. (eq:Q-dis<tr-dis<P-dis)}
The proof runs similar to the one of appendix \ref{Proof of Eq. (eq:p-dis<tr-dis<f-dis)}. Based on the completeness relation \cite{Serafini2017}
\begin{eqnarray}
\frac{1}{\pi} \int d^2\alpha \  \ketbra{\alpha}{\alpha} = \1
\end{eqnarray}
we define for simplicity $\E(\alpha) = \frac{1}{\pi} \ketbra{\alpha}{\alpha}$, i.e.~we have $Q_{\Rho}(\alpha) = \Tr{\E(\alpha) \Rho}$, as well as again $\Deltabm := \Rho_1 - \Rho_2 =  \sum_i \Delta_i \ketbra{i}{i}$. The first inequality is then obtained as follows,
\begin{eqnarray}
&\dis_\mathrm{kol}(Q_1, Q_2) = \frac{1}{2} \int d^2\alpha \ |Q_1(\alpha) - Q_2(\alpha)| = \frac{1}{2} \int d^2\alpha \ |\Tr{\E(\alpha)\Deltabm}| 
\\
&= \frac{1}{2} \int d^2\alpha \ |\sum_i \Delta_i \ubs{\geq 0}{\Tr{\E(\alpha) \ketbra{i}{i}}}| \leq \frac{1}{2} \int d^2\alpha \ \sum_i |\Delta_i| \Tr{\E(\alpha) \ketbra{i}{i}}
\\
&= \frac{1}{2} \sum_i |\Delta_i| \Tr{\int d^2\alpha \E(\alpha) \ketbra{i}{i}} = \frac{1}{2} \sum_i |\Delta_i| \ubs{=1}{\Tr{\1 \ketbra{i}{i}}} \\
&= \dis_\mathrm{tr}(\Rho_1, \Rho_2).
\end{eqnarray}

To prove the second inequality we use $\Rho_1 - \Rho_2 = \int d^2\alpha \ \big(P_1(\alpha) - P_2(\alpha)\big) \ketbra{\alpha}{\alpha}$ which yields
\begin{eqnarray}
&\dis_\mathrm{tr}(\Rho_1, \Rho_2) = \frac{1}{2} ||\Rho_1 - \Rho_2|| 
= \frac{1}{2} || \int d^2\alpha \ \big(P_1(\alpha) - P_2(\alpha)\big) \ketbra{\alpha}{\alpha} ||
\\
&\leq \frac{1}{2} \int d^2 \alpha \ |P_1(\alpha) - P_2(\alpha)| \cdot  || \ketbra{\alpha}{\alpha} || 
= \dis_\mathrm{kol}(P_1, P_2),
\end{eqnarray}
where we have again used properties of the trace norm. \hfill $\quad \square$

\section{Solution of the master equation}
\label{Equations of Motion for Driven Dumping}
The master equation is given by
\begin{eqnarray}
\frac{d}{dt} \Rho(t) \ &= \ \gamma_- \sin(\Omega t)\big( \A\Rho(t)\A^\dagger - \frac{1}{2} \{\A^\dagger \A, \Rho(t) \} \big)
\nonumber \\
&+ \ \gamma_+ \sin(\Omega t)\big( \A^\dagger\Rho(t)\A - \frac{1}{2} \{\A \A^\dagger, \Rho(t) \} \big),
\end{eqnarray}
and for any time independent operator $\hat{o}$ we have
\begin{eqnarray}
\frac{d}{dt} \braket{\hat{o}} = \frac{d}{dt}\Tr{\hat{o}\Rho} = \Tr{\hat{o}\frac{d}{dt}\Rho}.
\label{eq:Motion-O}
\end{eqnarray}
Furthermore, we know that for Gaussian states the Wigner function  $W(q,p)$ over the phase space with is defined by its displacement vector and covariance matrix
\begin{eqnarray}
d_i &= \braket{\bm{X}_i} \quad \und \quad \sigma_{ij} &= \braket{\bm{X}_i \bm{X}_j + \bm{X}_j \bm{X}_i} - 2\braket{\bm{X}_i}\braket{\bm{X}_j}
\end{eqnarray}
with $\X = (\q, \p)^T = \frac{1}{\sqrt{2}}\left(\A + \A^\dagger, -i (\A - \A^\dagger)\right)^T$. Accordingly we find
\begin{eqnarray}
d_0 &= \frac{1}{\sqrt{2}} \left(\braket{\A} + \braket{\A^\dagger}\right)
\quad \und \quad
d_1 &= \frac{-i}{\sqrt{2}} \left(\braket{\A} - \braket{\A^\dagger}\right)
\end{eqnarray}
as well as
\begin{eqnarray}
\sigma_{00} &= \var{\A} + \var{\A^\dagger} + 2\left( \braket{\A^\dagger \A} - \braket{\A^\dagger}\braket{\A} \right) + 1
\\
\sigma_{11} &= -\var{\A} - \var{\A^\dagger} + 2\left( \braket{\A^\dagger \A} - \braket{\A^\dagger}\braket{\A} \right) + 1
\\
\sigma_{01} &= \sigma_{10} = -i\left( \var{\A} - \var{\A^\dagger} \right)
\end{eqnarray}
with $\var{\hat{o}} = \braket{\hat{o}^2} - \braket{\hat{o}}^2$. Using (\ref{eq:Motion-O}) we now find
\begin{eqnarray}
\frac{d}{dt} \braket{\A}(t) &=& \gamma_-(t)\Tr{\A\A\Rho\A^\dagger - \frac{1}{2}(\A\A^\dagger\A\Rho + \A\Rho\A\A^\dagger)}
\nonumber \\
&+& \gamma_+(t)\Tr{\A\A^\dagger\Rho\A - \frac{1}{2}(\A\A\A^\dagger\Rho + \A\Rho\A\A^\dagger)} \nonumber \\
&=& \ \gamma_-(t)\Tr{\A^\dagger\A\A\Rho - \frac{1}{2}\A\A^\dagger\A\Rho - \frac{1}{2}\A^\dagger\A\A\Rho}
\nonumber \\
&+& \ \gamma_+(t)\Tr{\A\A\A^\dagger\Rho - \frac{1}{2}\A\A\A^\dagger\Rho - \frac{1}{2}\A\A^\dagger\A\Rho} \nonumber \\
&=& \ \frac{\gamma_-(t)}{2} \Tr{-[\A,\A^\dagger]\A\Rho} + \frac{\gamma_+(t)}{2}\Tr{\A[\A,\A^\dagger]\Rho} \nonumber \\
&=& \ -\frac{\gamma_-(t)-\gamma_+(t)}{2}\braket{\A}(t)
\end{eqnarray}
with $\gamma_\pm(t) = \gamma_\pm \sin(\Omega t)$ and in the last step we have used $[\A,\A^\dagger] = \1$. Similar calculations also lead to
\begin{eqnarray}
\frac{d}{dt} \braket{\A^\dagger}(t) &= -\frac{\gamma_-(t)-\gamma_+(t)}{2}\braket{\A^\dagger}
\\
\frac{d}{dt} \braket{\A^2}(t) &= -(\gamma_-(t)-\gamma_+(t))\braket{\A^2}
\\
\frac{d}{dt} \braket{(\A^\dagger)^2}(t) &= -(\gamma_-(t)-\gamma_+(t))\braket{(\A^\dagger)^2}
\end{eqnarray}
To obtain $\braket{\A^\dagger \A}$ we proceed as follows,
\begin{eqnarray}
\hspace{-18mm}
\frac{d}{dt}\braket{\A^\dagger \A} &=& \gamma_-(t) \Tr{\A^\dagger\A\A\Rho\A - \frac{1}{2}\A^\dagger\A\A^\dagger\A\Rho - \frac{1}{2}\A^\dagger\A\Rho\A^\dagger\A}
\nonumber \\
&+& \gamma_+(t) \Tr{\A^\dagger\A\A^\dagger\Rho\A - \frac{1}{2}\A^\dagger\A\A\A^\dagger\Rho - \frac{1}{2}\A^\dagger\A\Rho\A\A^\dagger}
\nonumber \\
&=& \ \gamma_-(t) \Tr{\A^\dagger \ [\A^\dagger,\A] \ \A\Rho}
\nonumber \\
&+& \ \gamma_+(t) \Tr{\A\A^\dagger\A\A^\dagger \Rho - \frac{1}{2}(\A\A^\dagger + [\A^\dagger,\A])\A\A^\dagger\Rho 
 - \frac{1}{2}\A\A^\dagger(\A\A^\dagger + [\A^\dagger,\A])\Rho}
\nonumber \\
&=& \ -\gamma_-(t)\Tr{\A^\dagger\A \Rho} + \gamma_+(t)\Tr{\A^\dagger\A\Rho + [\A,\A^\dagger]\Rho}
\nonumber \\
&=& \ -\gamma_-(t)\braket{\A^\dagger\A} + \gamma_+(t)\big(\braket{\A^\dagger\A} + 1\big).
\end{eqnarray}
With $\gamma_0 = \gamma_- - \gamma_+$ and $D(t) = \e^{\frac{\gamma_0}{\Omega}\big(\cos(\Omega t)-1\big)}$ these differential equations get solved by
\begin{eqnarray}
\braket{\A}(t) &= \braket{\A}_0 D(t)
\\
\braket{\A^\dagger}(t) &= \braket{\A^\dagger}_0 D(t)
\\
\braket{\A^2}(t) &= \braket{\A^2}_0 D(t)^2
\\
\braket{(\A^\dagger)^2}(t) &= \braket{(\A^\dagger)^2}_0 D(t)^2
\\
\braket{\A^\dagger\A}(t) &= \left(\braket{\A^\dagger\A}_0 - \frac{\gamma_+}{\gamma_0}\right) D(t)^2 + \frac{\gamma_+}{\gamma_0},
\end{eqnarray}
which finally yields
\begin{eqnarray}
\vec{d}(t) &= \vec{d}(0) D(t)
\\
\sigmabm(t) &= \sigmabm(0) D(t)^2 + \frac{2\gamma_+}{\gamma_0}\big(1-D(t)^2\big)\1 + \1.
\end{eqnarray}
Since $\sigmabm = \sigmabm_0$ is the covariance matrix of the Wigner function and $\sigmabm_P = \sigmabm_{1} = \sigmabm_{0} - \1$ one finally gets the covariance matrix of the P-function
\begin{eqnarray}
\sigmabm_P(t) &= \sigmabm_P(0) D(t)^2 + \frac{2\gamma_+}{\gamma_0}\big(1-D(t)^2\big)\1,
\end{eqnarray}
while the displacement is the same for Wigner function and P-function.
\newline
\newline

\def\newblock{\ }
\bibliography{Bibliography}

\end{document}